\documentclass[12pt,fleqn]{article}

\usepackage{graphicx}
\usepackage{setspace}
\usepackage{amsmath}
\usepackage{amsthm}
\usepackage{amsfonts}
\usepackage{multirow}
\usepackage{booktabs}
\usepackage{amssymb}
\usepackage{authblk}
\usepackage{color}
\usepackage{url}
\usepackage[comma]{natbib}
\usepackage{authblk}
\usepackage{float}

\newcommand{\dv}{\mbox{\boldmath$d$}}

\newcommand{\gav}{\mbox{\boldmath$\gamma$}}

\newcommand{\zv}{\mbox{\boldmath$z$}}
\newcommand{\BF}{{\rm BF}}
\newcommand{\BFDR}{{\rm BFDR}}
\newcommand{\FDR}{{\rm FDR}}
\newcommand{\lv}{{\bf 1}}
\newcommand{\lamv}{\mbox{\boldmath$\lambda$}}
\newtheorem{theorem}{THEOREM}
\newtheorem{definition}{DEFINITION}
\newtheorem{prop}{PROPOSITION}

\newtheorem{algorithm}{ALGORITHM}

\begin{document}
\title{\large \bf A Unified View of False Discovery Rate Control: Reconciliation of Bayesian and Frequentist Approaches }
\author {Xiaoquan Wen\thanks{xwen@umich.edu}}
\affil{University of Michigan}
\date{}
\maketitle

\abstract{
This paper explores the intrinsic connections between the Bayesian false discovery rate (FDR) control procedures and their counterpart of frequentist procedures. 
We attempt to offer a unified view of FDR control within and beyond the setting of testing exchangeable hypotheses. 
Under the standard two-groups model and the oracle condition, we show that the Bayesian and the frequentist methods can achieve asymptotically equivalent FDR control at arbitrary levels.   
Built on this result, we further illustrate that rigorous post-fitting model diagnosis is necessary and effective to ensure robust FDR controls for parametric Bayesian approaches.
Additionally, we show that the Bayesian FDR control approaches are coherent and naturally extended to the setting beyond testing exchangeable hypotheses.
Particularly, we illustrate that $p$-values are no longer the natural statistical instruments for optimal frequentist FDR control in testing non-exchangeable hypotheses.
Finally, we illustrate that simple numerical recipes motivated by our theoretical results can be effective in examining some key model assumptions commonly assumed in both Bayesian and frequentist procedures (e.g., {\em zero assumption}). 
}

\newpage

\section{Introduction}

Ever since the seminal work by Benjamini and Hochberg \citep{Benjamini1995}, false discovery rate (FDR) control has become increasingly popular in the scientific practice of multiple hypothesis testing. 
Comparing to the alternative statistical strategies, such as the control of family-wise error rate (FWER), the FDR control is much more appealing in large-scale data analyses due to its cost-effectiveness considerations. 
In the last decade, both the theory and the computational algorithms for FDR control have matured in the setting of testing multiple exchangeable hypotheses, thanks to the important works of  \cite{Benjamini1995, Efron2001, Storey2002, Storey2003, Storey2004, Genovese2004, Newton2004, Mueller2004, Ghosh2006, Sun2007, Whittemore2007, Storey2007, Efron2008, Muralidharan2010, Stephens2016}, just name a few.
In many fields, e.g., genetics, genomics, and molecular biology, FDR control has emerged as the method of choice to quantify uncertainty and safeguard potential type I errors in scientific discovery processes. 

The existing FDR control approaches can be classified into two distinct categories.  The first category is represented by the frequentist FDR control approaches, e.g., the Benjamini-Hochberg (B-H) procedure and the Storey's $q$-value method, which are characterized by their use of $p$-values to estimate and control FDR.
The methods in the second category require explicit computation of posterior probabilities, also commonly known as local fdr's \citep{Efron2001}, and are considered as Bayesian FDR control approaches. 
The two categories of the methods are fundamentally different because the underlying quantities that measure the expected false discovery proportion are conceptually distinct (the details are explained in Section 2). 
Both categories of approaches have their own strength and weakness. 
For example, the frequentist FDR control methods, which requires no explicit specification of the alternative models, are known to be robust to the misspecification of the alternative models; whereas the Bayesian methods, especially the parametric approaches, have great flexibility to incorporate ancillary data/prior information that can significantly improve the power of hypothesis testing.

There have been many efforts to bridge the gap between the two types of approaches in the literature: Storey \citep{Storey2002} shows that the Bayesian interpretation of the $q$-value method and Efron \citep{Efron2008} demonstrates the connections between Bayesian FDR/local fdr and frequentist FDR under some specific power functions.  
In this article, we aim to explore a more in-depth connection between the two types of approaches. 
In particular, we are interested in investigating if the Bayesian and frequentist FDR control procedures can yield {\em concordant} results for analyzing the same data set.
If the answer is yes, we would like to identify the corresponding necessary and sufficient conditions.
Answers to these questions have important implications for the best practice of  FDR controls for both frequentist and Bayesian approaches.  
Beyond the standard setting characterized by the two-groups model \citep{Efron2008}, we further explore the extensibility an optimality of the two types of FDR control methods in some more complex, but increasingly more realistic, scenarios.

The paper is structured as the follows. We first provide necessary background on frequentist and Bayesian FDR controls, then proceed to present our main theoretical results. Subsequently, we provide numerical illustrations of our theoretical results. Finally, we end the paper by the summarizing the main conclusions and a brief discussion on future directions.

\section{Background}

\subsection{Model and notation}

We follow \cite{Genovese2004, Efron2008} to adopt a mixture model description of testing $m$ exchangeable hypotheses as our starting point.
Let the latent binary indicator $\gamma_i = 1$  denote that the observed data of of the $i$-th hypothesis are generated from the alternative scenario,  and 0 otherwise. 
We use $z_i$ to represent the observed data/summary statistic from the $i$-th test, which can be computed independent of data from all other tests.
Note that $z_i$ can be a $p$-value in a frequentist approach or a Bayes factor (i.e., a marginal likelihood ratio statistic) in a Bayesian approach; it can also be a vector \citep{Stephens2016} instead of a single number.   
We denote the collection of the observed data and the corresponding latent indicators by $\zv = (z_1, \dots, z_m)$ and $\gav = (\gamma_1, \dots, \gamma_m)$, respectively.

The assumed data generation mechanism can be described by the following probabilistic model,
\begin{equation}
 \begin{aligned}
    &  \gamma_i, \dots, \gamma_m \mid \pi_0 \stackrel{i.i.d}{ \sim } {\rm Bernoulli}\, (1 - \pi_0), \\
    &  z_i \mid \gamma_i = 0 \, \sim f_0, \\
   &  z_i \mid \gamma_i = 1 \, \sim f_1.
 \end{aligned}
\end{equation}
In particular, we assume the independence of $\gamma_i$'s conditional on the hyper-parameter $\pi_0$. 
The model also implies that $z_i$'s are exchangeable, and we will relax this particular assumption later in the paper.
This particular mixture model is known as the  \emph{two-groups model} \citep{Efron2008}. 
It is also a special case of conditional independent hierarchical model (CIHM) discussed in \cite{Kass1989}.
Unless otherwise specified,  we assume that  the theoretical null, $f_0$, is known and correctly specified while both $f_1$ and $\pi_0$ are unknown.
Finally, we use $f_c$ to denote the mixture distribution $\pi_0 f_0 + (1- \pi_0) f_1$.

\subsection{FDR estimation and control}

Let $\delta_i(\cdot)$ denote a binary decision function with respect to the $i$-th hypothesis, i.e., $\delta_i = 1$ indicates the rejection of the $i$th hypothesis, and 0 otherwise.  
Under the two-groups model and given a decision function, we define the false discovery proportion (FDP), a random variable, by 
\begin{equation}
  \begin{aligned}
  {\rm FDP} &= \frac{\mbox{number of falsely rejected null hypothese}}{\max(1, \mbox{number of total rejections})} \\
            & = \frac{\sum_{i=1}^m  \delta_i \, (1-\gamma_i)}{ \left( \sum_{i=1}^m  \delta_i \right) \,   \vee \, 1}.
  \end{aligned}
\end{equation}
Both Bayesian and frerquentist FDR control procedures aim to bound some type of expectation of FDP.

\subsubsection{Bayesian FDR and the optimal control procedure}

The Bayesian FDR is defined as the conditional expectation of FDP given the observed data. Because of the explicit conditioning,  $\delta_i$'s are considered deterministic.
\cite{Mueller2006} shows that the optimal decision rule to control Bayesian FDR is given by
\begin{equation} \label{bayes.dec.rule}
  \delta_{B,i}(t_b) = \lv_{ \{ u_i \, \le\,  t_b  \}},
\end{equation}
where $u_i$ denotes the posterior probability $\Pr(\gamma_i = 0 \mid \zv) $ and $t_b$ represents a pre-defined threshold on $u_i$'s. 
The optimal Bayesian FDR control procedure rejects the null hypotheses for which the individual posterior probabilities of false discovery, i.e., $u_i$'s,  are small. 
It straightforwardly follows from the definition that
 \begin{equation}\label{bfdr}
\BFDR(t_b)  = {\rm E}( {\rm FDP} \mid \zv ) = \frac{\sum_{i=1}^m u_i \, \delta_{B,i}(t_b)}{ \left[\sum_{i=1}^m \delta_{B,i}(t_b) \right] \, \vee \, 1}.
 \end{equation}
In the literature, the posterior probability of false discovery is also commonly referred to as {\em local fdr}, a terminology coined by Efron.

\subsubsection{Frequentist FDR control}

The frequentist FDR is defined as the unconditional expectation of FDP, i.e.,
\begin{equation}\label{ffdr}
  {\rm FDR} = {\rm E}( {\rm FDP} ) = {\rm E} ( \BFDR),
\end{equation}
where, in comparison to Bayesian FDR, an additional expectation is taken with respect to the hypothetically re-sampled data, $\zv$. 
Widely used Benjamini-Hochberg (B-H) and Storey's $q$-value algorithms are designed to control frequentist FDR. 
The decision rules embedded in those algorithms have the common form, 
\begin{equation}
  \delta_{F, i}(t_f) = \lv_{\{ p_i  \, \le \, t_f \}},
\end{equation}
where $p_i$ and $t_f$ represent the $p$-value of test $i$ and a pre-defined threshold, respectively.  
While the exact evaluation of frequentist FDR (\ref{ffdr}) is difficult for a given $t_f$, \cite{Genovese2004} shows an accurate asymptotic approximation, 
\begin{equation}\label{ffdr.exp}
  {\rm FDR}(t_f) = \frac{m \, \pi_0 \,  t_f}{\left[ \sum_{i=1}^m \delta_{F,i}(t_f) \right] \, \vee \, 1 } + o \left( \frac{1}{m} \right).
\end{equation}
%where $\sum_{i=1}^m \delta_{F,i}(t_f) / m$ is a plug-in estimator of the cumulative distribution function (CDF) of the mixture distribution evaluated at $t_f$.
The approximation, $\frac{m \, \pi_0 \,  t_f}{\left[ \sum_{i=1}^m \delta_{F,i}(t_f) \right] \, \vee \, 1 }$, is also referred to as an estimate of {\em Fdr} by \cite{EfronBook}.
In \cite{Sun2007} the same quantity is refers to as {\em marginal FDR} (mFDR).

The B-H and the $q$-value procedures are two most commonly applied frequentist FDR control methods.
The difference between the two lies in their treatment of unknown parameter $\pi_0$. 
The $q$-value procedure plugs in a non-parametric estimator of $\pi_0$ that is derived based on sample quantile information.
In comparison, the B-H procedure simply sets $\pi_0 = 1$ in (\ref{ffdr.exp}).
As a result, its estimate essentially becomes an upper-bound of desired FDR.

\subsection{Rejection path and control of FDR}

In this paper, we propose a new statistical device to perform comparisons between FDR control procedures. We formally define a {\em rejection path} as follows.
\begin{definition}
Given a decision rule of a FDR control procedure,  $\delta_i = \lv_{\{ s_i \le t \}}$,  
the rejection path for a data set is the sequence of estimated (Bayesian or frequentist) false discovery rates sequentially evaluated at the threshold values determined by the order statistics, $s_{(1)},\dots,s_{(m)}$, i.e., 
\begin{equation*}
    \FDR(t=s_{(1)}), \FDR(t=s_{(2)}), \dots,\FDR(t=s_{(m)}).
\end{equation*} 
\end{definition}

The rejection path contains the complete information required to control FDR at any pre-defined level $\alpha$. 
Specifically, we find 
\begin{equation*}
  l = \arg\max_k \{k: \FDR(s_{(k)}) \le \alpha \},
\end{equation*}
and reject the tests corresponding to $s_{(1)},...,s_{(l)}$.

Different FDR control procedures typically yield distinct rejection paths when analyzing a given data set. If the rejection paths from two procedures are always identical, the following two conditions are necessarily satisfied.
First, the two procedures rank individual tests in the same order.  
In addition,  the two procedures yield the same FDR estimates at all threshold values in the rejection path. 
The ranking of individual tests is typically determined by the choice of test statistics and can be examined by the method like receiver operating characteristic (ROC) curves in simulation settings.
The evaluation of FDR is the one of the main interests of this paper. 

Direct comparison of two FDR control procedures yields quantitative assessment of relative conservativeness of the different FDR estimation approaches. 
Figure \ref{rejpath.demo} illustrates the pair-wise comparison of the rejection paths from the B-H procedure and the $q$-value procedure for the Hedenfalk data distributed in the R package {\it q-value}. 

\begin{figure}
\centering
\vspace{0pt}
\includegraphics[ width=0.5\textwidth]{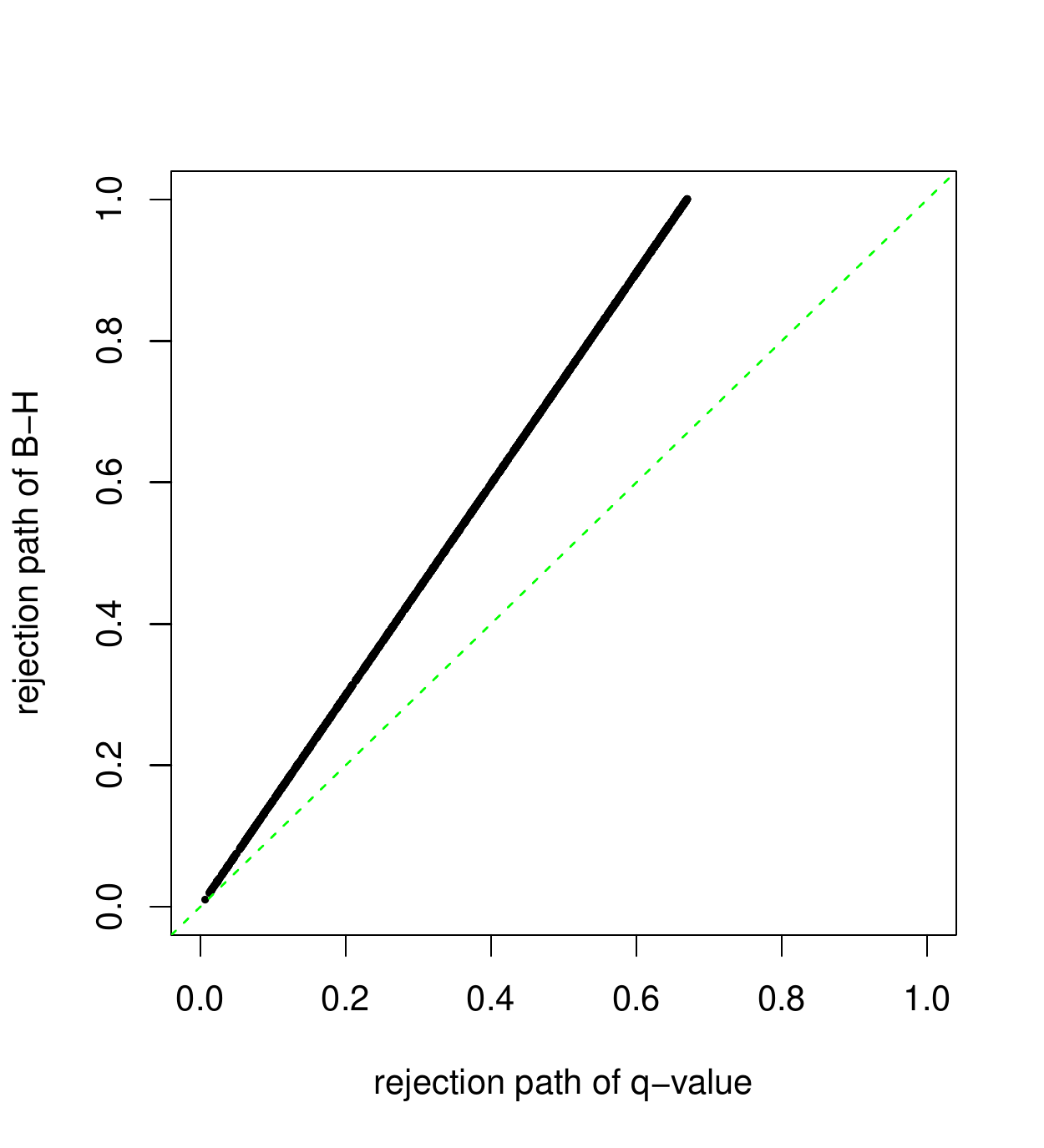}
\caption{{\bf pairwise comparison of rejection paths of the $q$-value method and the B-H procedure.} We  use the Hedenfalk data distributed with the R package {\it $q$-value}, which contain 3,170 $p$-values. The comparison of the rejection paths indicates that the B-H procedure always overestimate the FDR than the $q$-value procedure by a linear factor.  \label{rejpath.demo}}
\end{figure}

\section{Results}

\subsection{Frequentist property of Bayesian FDR control}

Our first new result concerns the quantitative connection between the Bayesian and frequentist FDR control procedures in a setting where $m$ is large.
We consider an ideal scenario where both $\pi_0$ and $f_1$ are correctly specified, which corresponds to an oracle condition  \citep{Sun2007}.
Under such setting, the posterior false discovery probability for each test, $u_i^* = \Pr( \gamma_i = 0 \mid \pi_0, z_i)$,  can be analytically computed by 
\begin{equation}\label{bf2pip}
\begin{aligned}
  u_i^*  & = \frac{\pi_0 f_0(z_i)}{\pi_0 f_0 (z_i) + (1-\pi_0) f_1 (z_i)}\\
         &  = \frac{\pi_0}{\pi_0 + (1-\pi_0) \BF^*(z_i)},
 \end{aligned}
\end{equation}
where $\BF^*(z_i) := f_1(z_i)/f_0(z_i)$ denotes the marginal likelihood ratio statistic/Bayes factor.

We show that under specific conditions, Bayesian and frequentist FDR control procedures are asymptotically equivalent. 
The main result is summarized in Theorem 1.
 \begin{theorem} Assume that 
 \begin{enumerate}
  \item $\pi_0$, $f_0$ and $f_1$ are correctly specified;
  \item  There exists a monotone mapping from the marginal likelihood ratio statistics to the $p$-values.
 \end{enumerate}
It follows that 
\begin{equation*}
 \BFDR \left (t_{b,i} =  u^*_{(i)} \right)   \xrightarrow{a.s.}  \FDR \left(t_{f,i} =  p_{bf, (i)} \right) \mbox{ for all } i.
\end{equation*}
\end{theorem}

\begin{proof}
  See Appendix A.
\end{proof}

\noindent {\bf Remark 1.} The existence of monotone mapping from the Bayes factors to the $p$-values implies that the monotonic relationship between the two statistical instruments holds for any observed data sets. 
The most straightforward way to guarantee this condition is to derive  $p$-values directly from the corresponding likelihood ratio statistics.
In this scenario, Bayesian and frequentist analyses make exactly the same modeling assumptions and extract exactly the same information from data, and the difference in results should solely reflect their procedural difference in FDR controls.  
Nevertheless, it should be noted that the monotonic relationship between Bayes factors and $p$-values can hold in much more general settings. 
For example,  Bayesian hypothesis testing and inference based on $z$-statistics (or $z^2$-statistics) has been widely studied and applied in a wide range of applications (e.g., genetic association studies).  
It is commonly assumed that $z_i  \sim  {\rm N}(0,1)$ under the null and  $z_i \sim {\rm N}(0, 1+k)$ (for some $k > 0$) under the alternative. 
(Note that the alternative model is sometimes interpreted as a random effect assumption.)
Under this setting, both \cite{CoxBook} and \cite{Wakefield2009} have shown that
\begin{equation} \label{abf}
   \BF(z_i) = \frac{1}{\sqrt{1+k}} \exp \left( \frac{1}{2} \, \frac{k}{1+k} \, z_i^2 \right), 
\end{equation}    
which is a  monotonic transformation of $z_i^2$. 
As a consequence, the $p$-value of the Bayes factor (\ref{abf}) is identical to the two-sided $p$-value of the corresponding $z$-statistic.  
More generally, the above relationship holds for modeling of $z$-statistics under alternatives using mixtures of normal and/or uniform distributions under the uni-modal alternative (UA) assumption that is recently proposed by \cite{Stephens2016}.  

\noindent {\bf Remark 2.} 
Both \cite{Mueller2006} and \cite{Sun2007} have shown that the oracle Bayesian procedure is optimal to minimize false non-discovery rate (FNR, a quantity measuring type II errors) while controlling FDR.
Our result essentially extends the optimality results to the frequentist FDR control procedures due to the asymptotic equivalence.

\noindent {\bf Remark 3.} The key assumption on the correctness of model specification is critical to ensure the convergence result. 
Although it is unrealistic to expect this particular assumption to hold \emph{exactly} in practice, the high-dimensional setting of practical applications typically yields a large amount of informative data and offers an opportunity for accurate model estimation. 
Additionally, it underscores the necessity and importance of careful model specification and thorough diagnosis in the practice of Bayesian FDR control. 
This topic will be extensively discussed in the subsequent sections.

Theorem 1 goes beyond the previous results by \cite{EfronBook} and \cite{Storey2003}, both of which focused on the Bayesian representation/interpretation of the frequentist FDR estimation, i.e., the functional form of equation (\ref{ffdr.exp}).  
It highlights the point-wise convergence of rejection paths of the Bayesian and frequentist procedures under the oracle condition. Thus, in a high-dimensional setting and under correct model specifications, the Bayesian and frequentist FDR control procedures are expected to be highly concordant for {\em all} pre-defined FDR control levels.
Figure \ref{prop1.demo} provides a numerical illustration of Theorem 1 using a set of simulated data, where good numerical concordance between the Bayesian and frequentist procedures can be observed when $m \sim 10^3$.

\begin{figure}
\centering
 \begin{tabular}[t]{c c c}
    \includegraphics[width=.32\textwidth]{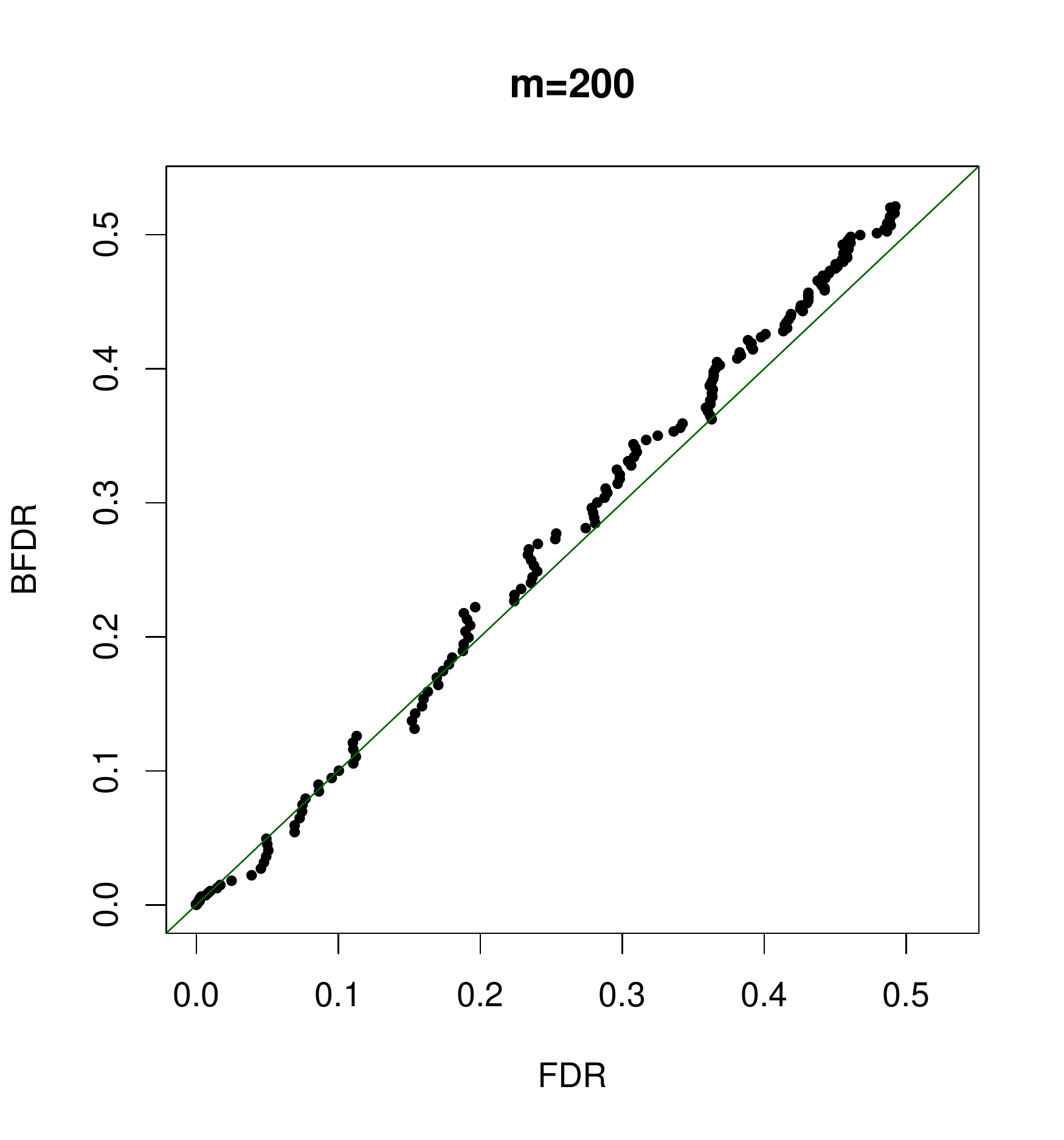}  &    \includegraphics[width=.32\textwidth]{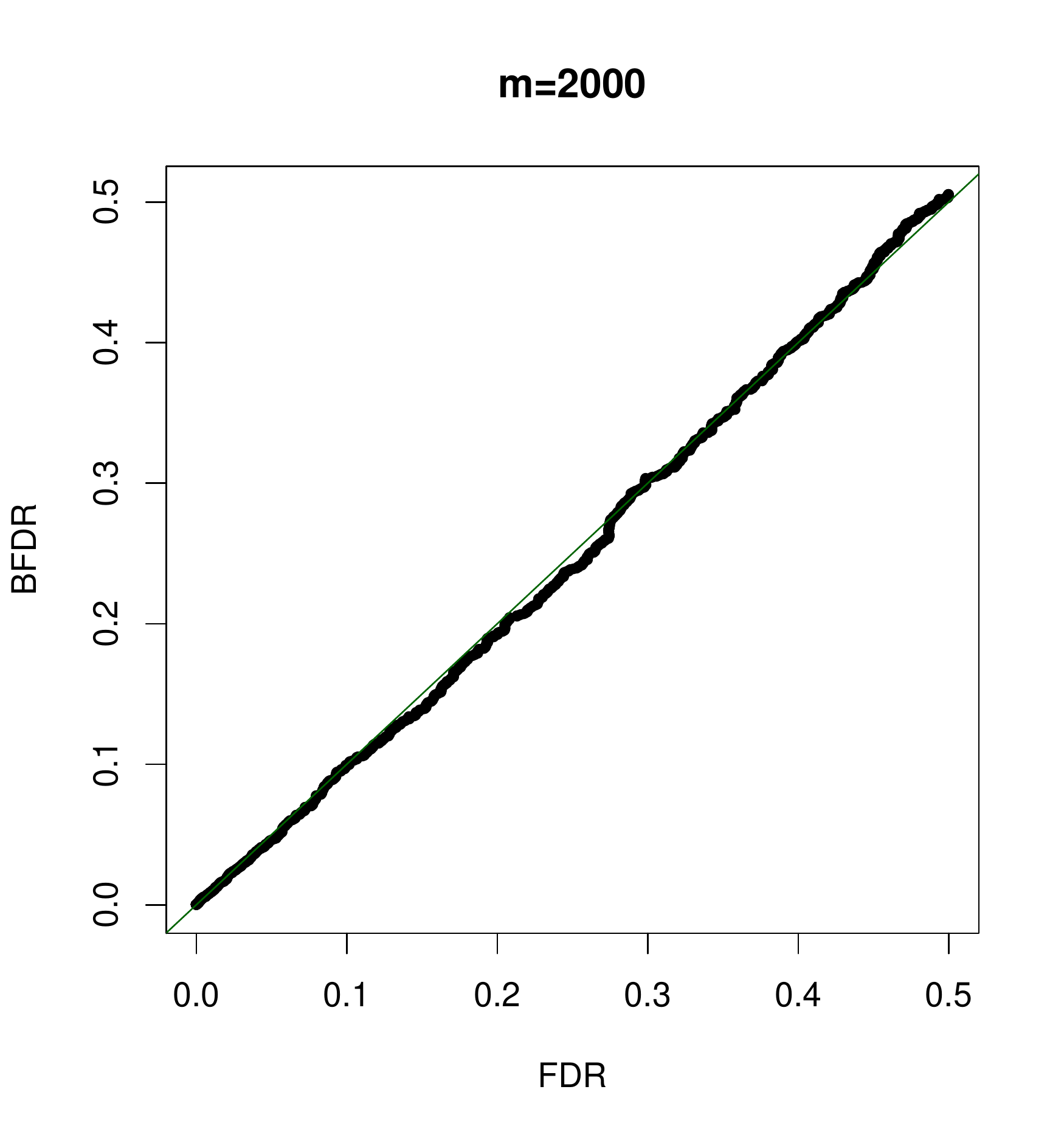} &
   \includegraphics[width=.32\textwidth]{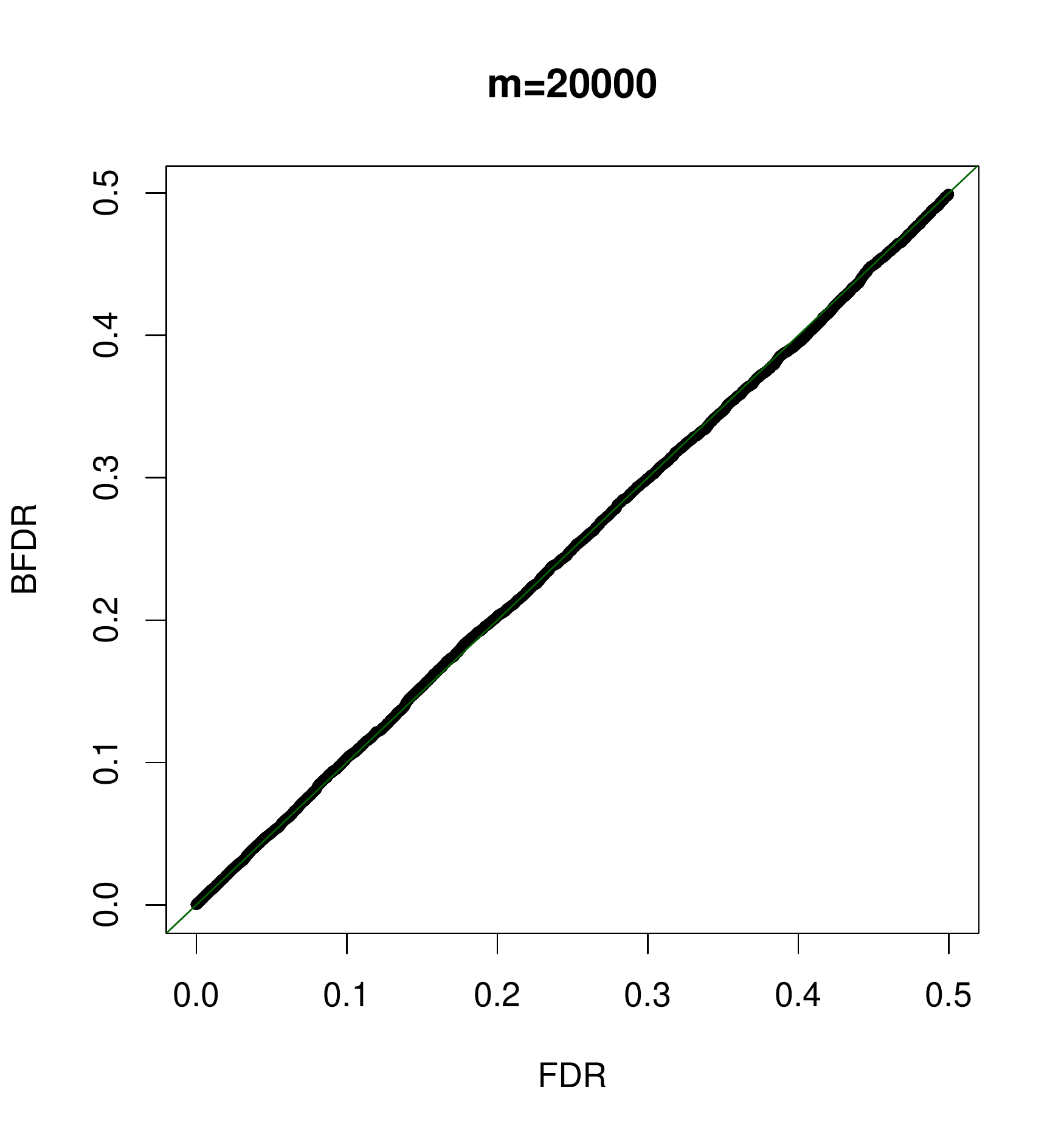}  \\
\end{tabular}   
\caption{{\bf pairwise comparison of rejection paths from the Bayesian and frequentist FDR control procedures under the oracle condtion.}  In all three panels, $z$-scores are drawn from the mixutre distribution, $\pi_0 N(0,1) +(1-\pi_0) N(0, 1+k)$, where we set $\pi_0 = 0.5$ and $k=10$. Each panel represents a simulation with $m$ ($ = 200, 2000$ and 20,000, respectively) independent tests.\label{prop1.demo}}
\end{figure}

\subsection{Approximate inference of  Bayesian FDR}

Standard Bayesian inference requires parametric specification of $f_1$ by $f_{\theta}$, where $\theta$ denotes the necessary hyper-parameters.
Ideally, $f_{\theta}$ should be flexible enough such that it can accurately approximate the true $f_1$.
The examples of such kind are illustrated by \cite{Stephens2016}.   
To enable Bayesian computation, the priors for the set of all the hyper-parameters $\lambda := (\pi_0, \theta)$ also need to be specified.

One of the factors that hinder the usage of the standard Bayesian inference for FDR control in practical settings is its expensive computational cost due to lacking of general analytic solutions for computing local fdr's in a general CIHM. 
Markov chain Monte Carlo (MCMC) algorithms are theoretically possible, however, the lack of scalability and/or effective convergence diagnosis often makes applying Bayesian procedure in high-dimensional settings less appealing. 
Here, we advocate the use of parametric empirical Bayes (PEB) procedure to fit CIHM and argue that PEB  generally provides adequate approximations to exact Bayesian inference in high-dimensional hypothesis testing settings.
  %The use of EB procedure to control FDR is best known as the local FDR (lfdr) control procedure, which is made popular by the works of  \cite{}.  
Comparing to the full Bayesian inference, the key distinction in the PEB procedure is to replace the desired local fdr, $\Pr(\gamma_i = 0 \mid \zv)$, with an empirical Bayes estimate, namely  $\Pr(\gamma_i=0 \mid  \hat \lambda, \, z_i  )$, where $\hat \lamv$ denotes the MLE of $\lambda$ and is a function of all observed data $\zv$.  
This procedure is justified by the result of \cite{Kass1989}, which shows that the PEB approach can be regarded as a \emph{computational approximation} of the full Bayesian computation in the setting of CIHM (for reasonably specified, yet arbitrary,  priors on $\lambda$). 
More specifically,  Result 1 of \cite{Kass1989} implies that, with a pre-specified alternative model $f_\theta$, 
\begin{equation}\label{eb.appx}
  \Pr(\gamma_i = 0 \mid  \zv ) = \Pr( \gamma_i = 0 \mid z_i, \hat \lambda) \cdot \left[ 1 + o ( \frac{1}{m}) \right],
\end{equation}
which is obtained by applying Laplace's method and re-normalizing the resulting approximation into a valid probability distribution. 

The PEB approximation is particularly attractive for at least two reasons. First,  its result has a \emph{relative} error of $o(\frac{1}{m})$ which is appealing for approximating small to modest posterior probabilities. 
In comparison, traditional Monte Carlo based approach only ensures the control of \emph{additive} error at the same scale. 
This point is critically important because the accurate evaluation of small to modest posterior probabilities is the key to ensuring accurate Bayesian FDR control.
Second, the EB procedure, which can be carried out by some well-established optimization algorithms (e.g., EM), is much more computationally efficient than the  MCMC fitting of CIHM and is particularly suitable in a big data setting.
It should be clear that the overall effectiveness of EB-based approximate inference procedure is dictated by the informativeness of the data with respect to the key model parameters. 
In this sense, the high-dimensionality of the simultaneous tests creates a nearly ideal asymptotic setting for accurate parameter estimation. 

The PEB also produces an estimate of the mixture density $f_c$, i.e., 
\begin{equation}
  \hat f_c = \hat \pi_0 f_0 + (1-\hat \pi_0) f_{\hat{\theta}}.
\end{equation}
Interestingly, by the same argument of \cite{Kass1989}, it can be shown that $\hat f_c$ is an approximation of the density of the posterior predictive distribution,  $f (z_{\rm new} \mid \zv) $, i.e., 
\begin{equation}\label{ppd.appx}
  f (z_{\rm new} \mid \zv) = \int f(z_{\rm new} \mid  \lambda) f( \lambda \mid \zv) d\, \lambda = \hat f_c(z_{\rm new}) \left [ 1 + o(\frac{1}{m}) \right]. 
\end{equation}
The posterior predictive distribution is instrumental for Bayesian model diagnosis, which is critical to ensure robust Bayesian FDR control.
The usage of $\hat f_c$ for Bayesian model checking will be illustrated in the next section.

\subsection{Robustness and model diagnosis}

Model misspecification and inadequate fitting can lead to either conservative or anti-conservative behaviors in FDR control.
While the former (e.g.,  overestimating $\pi_0$) is typically accepted, the latter (e.g.,  underestimating $\pi_0$) inflates FDR and is considered dangerous.
It is therefore critical to ensure the robustness of the FDR control by either theoretical arguments and/or practical model diagnosis procedures.
In particular, we expect an effective diagnosis approach could not only flag improper modeling assumptions and/or inadequate fitting but also provide a \emph{directional} indication of anti-conservative FDR control results. 

% Frequentist robustness
The frequentist FDR control procedures mainly rely on the $p$-values derived from the null hypothesis and seemingly avoid the detailed specifications of alternative models. 
Nevertheless, a \emph{precise} frequentist FDR estimation requires an accurate estimate of $\pi_0$, which inevitably needs the knowledge of the alternative model.
A common strategy to circumvent this difficulty is to use an upper bound estimate of $\pi_0$ and acknowledge that the conservativeness of the resulting procedure.  
For example, the B-H procedure simply plugs in the most natural upper bound by setting $\pi_0 = 1$ in Equation (\ref{ffdr}). 
The $q$-value procedure takes advantages of the high-dimensional setting of multiple hypothesis testing and finds a tighter upper bound estimate for $\pi_0$ based on the sample quantile information. 
More specifically, for any $\eta \in (0,1)$ and noting that $p$-values are uniformly distributed under the null model, it follows from the Kolmogorov's strong law of large numbers that
\begin{equation}  
 \label{qlln}
 \frac{1}{m} \sum_{i = 1}^m \lv \{p_i \ge  \eta \}  \xrightarrow{a.s.}   \pi_0 \, (1 - \eta) + (1-\pi_0) \, [1 - F_{p^*}(\eta)],
\end{equation}
where $F_{p^*}(\cdot)$ denotes the CDF of the $p$-values under the (unknown) alternative model.
By ignoring the second term on the right-hand side,  an upper bound estimate of $\pi_0$ is obtained by 
\begin{equation}\label{q.est}
  \tilde \pi_0 = \min \left(1, \frac{\sum_{i=1}^m  \lv \{p_i \ge \eta \}}{m\, (1-\eta)} \right),
\end{equation}
and it satisfies the upper bound condition in the sense,
\begin{equation} \label{upper.bound}
   \tilde \pi_0  \ge  \pi_0 \mbox { almost surely}, \mbox{ as } m \to \infty.
\end{equation}
Theoretically, the upper bound property of $\tilde \pi_0$ holds for arbitrary $\eta$ values, although setting $\eta$ close to 0.50 reduces the variance of the estimate. 
In practice, some use $\tilde \pi_0$ extrapolated at $\eta \to 1$ and argue it provides a tighter bound assuming that most of large $p$-values are generated from the null hypothesis.
However, this assumption can be false if the alternative scenario assumed by the unimodal alternative (UA) assumption is indeed true, where the data generated from the alternative models can be very much likely to produce large $p$-values  (see discussions in \cite{Stephens2016}).   
Nevertheless, in all cases, $\tilde \pi_0$ should provide a reliable upper bound estimate and ensures the robust performance of FDR control. 
Finally, we note that estimator (\ref{q.est}) is not restricted to the usage of $p$-values and can be generalized to almost any test statistics with available quantile information from the null distribution. 
For example, if the $p$-values are derived from the two-sided tests of $z$-statistics, the following estimator based on the $\eta$-quantile of $z^2$, denoted by $\chi^2_\eta$,  provides the identical estimation as (\ref{q.est}),
$$ \tilde \pi_0  =  \min \left(1, \frac{\sum_{i=1}^m  \lv \{z^2_i \le \chi^2_\eta \}}{m\, \eta} \right). $$

The {\em locfdr} procedure \citep{Efron2001, Efron2008}  utilizes a non-parametric empirical Bayes (NPEB) inference framework, where the density of the mixture distribution, $f_c$,  is estimated non-paramterically.
To estimate $\pi_0$,  it makes the ``zero assumption"  (ZA), i.e., $f_1(z) \ll f_0(z) $ as $z \to 0$ in case of simple $z$-statistics.
Similar to the case of estimator (\ref{q.est}), we view the ZA as a robust assumption for FDR control: the violation of the assumption results in over-estimated $\pi_0$.
In general, {\it locfdr} ensures an accurate fit of $f_c$, if the  ZA holds, the resulting estimate is accurate and the subsequent FDR control is optimal or near-optimal; otherwise, it overestimates $\pi_0$ and results in conservative control of FDR. 

In comparison, the standard Bayesian inference with the parametric specification of $f_\theta$ is more susceptible to model misspecification and/or inadequate fitting. 
Here we emphasize the critical importance of conducting post-fitting model diagnosis utilizing the posterior predictive distribution. 
In the case that Bayesian inference is approximated by the PEB procedure, the standard posterior predictive checking in Bayesian inference converges to the common practice of examining the goodness-of-fit of $\hat f_c$,  based on the argument provided by equation (\ref{ppd.appx}).
In particular, we find that it is effective to examine the alignment between the sample quantiles from the observed data and the corresponding theoretical quantiles computed (or estimated) from $\hat f_c$, which is an idea naturally extended from \cite{Cook2006}.  
In some of the examples that we show in Section 4,  the patterns of significant misalignment of the two quantiles also provide an important indication if a particular inadequate fit leads to a conservative or anti-conservative FDR control. 

To provide a theoretical guidance to Bayesian model diagnosis, we extend Theorem 1 to a  practical setting where the inference of  $\pi_0$ and $f_1$ is required. 
Specifically, we consider the estimates $\hat \pi_0$ and $\hat f_c$ (or equivalently, $f_{\hat \theta}$) are obtained through an inference procedure (e.g., PEB), and the local fdr for test $i$ is evaluated by $\hat u_i: = \hat u (z_i; \hat \pi_0) = \frac{\hat \pi_0 f_0 (z_i)}{\hat f_c(z_i)}$.
Furthermore, we denote  the induced order statistics from $\{\hat u_1, ..., \hat u_m\}$ by $( \hat u_{(1)}, \hat u_{(2)}, ..., \hat u_{(m)})$.   
Our main result is summarized in Theorem 2.
 \begin{theorem}
  Assume that
  \begin{enumerate}
   \item there exists a monotone mapping from the marginal likelihood ratio statistics, $f_0/\hat f_1$, to the corresponding $p$-values;
   \item ${\rm E}\, || \hat f_c - f_c ||^2 \to 0 $. 
  \end{enumerate}
It follows that 
\begin{equation*}
 \BFDR \left (t_{b,i} =  \hat u_{(i)} ; \hat \pi_0 \right)  \xrightarrow{p} \FDR \left(t_{f,i} =  p_{ (i)}; \hat \pi_0 \right) \mbox{ for all } i.
 \end{equation*}
 
Furthermore, if   $\hat \pi_0 \xrightarrow{p} \pi_0$,   
 \begin{equation*}
      \BFDR \left (t_{b,i} =  \hat u_{(i)} ; \hat \pi_0 \right)  \xrightarrow{p}  \BFDR \left (t_{b,i} =  u^*_{(i)} \right)
 \end{equation*}
\end{theorem}
\begin{proof}
  See Appendix B.
\end{proof}

\noindent {\bf Remark.} Because the identifiability issue associated with mixture model inference can be extremely complicated, the overall measure of adequate fitting represented by the second assumption of Theorem 2 does not generally guarantee that $\hat \pi_0$ is accurat in theory.  
Thus, our convergence result between $\BFDR$ and $\FDR$ is conditional on the $\hat \pi_0$ estimate, and the accuracy of the $\pi_0$ estimate needs to be further ensured to establish that the FDR control results are approaching the theoretical optimal.  

Theorem 2 implies that if a Bayesian model is adequately fit by the parametric empirical Bayes approach, we expect a strong agreement between the Bayesian and the corresponding frequentist FDR control procedures, provided that the same estimate of $\hat \pi_0$ is applied to both procedures.
Conversely, notable discordance between the Bayesian and frequentist approaches indicates poor modeling assumption and/or inadequate fitting of the assumed parametric model.

Note that Theorem 2 is not directly applicable to the NPEB approach implemented in the {\it locfdr} method, despite that the overall adequate fitting of the mixture distribution, $f_c$, is usually achieved.  
This is because the Bayes factors/likelihood ratio statistics are implicit (i.e., no explicit form of $f_\theta$ is estimated), and the theoretical validation of the first assumption is not plausible. 
Nevertheless, in practice,  the rank correlation between local fdr's from the NPEB approach and the corresponding p-values can be numerically examined. 
As the rank correlation $\to 1$, it is expected, based on Theorem 2,  that the two approaches yield concordant FDR control results. 
This point will be fully illustrated in Section 4.

One of the important practical implications of Theorem 2 is that the goodness-of-fit of the mixture distribution (i.e., the second assumption) is a {\em necessary} condition to ensure the frequentist property of the Bayesian FDR control. 
This result underscores the importance of model checking as a critical part of an overall  Bayesian FDR control procedure.
Beyond the scope of Theorem 2,  exploring the detailed patterns of misfitting by $\hat f_c$ may offer valuable insights in determining whether the inadequate fit leads to conservative or anti-conservative FDR control.

\subsection{Multiple hypothesis testing with non-exchangeable data}

In many practical settings of multiple hypothesis testing, there exist ancillary data that make the observed data non-exchangeable. 
For example, in the analysis of microarray/RNA-seq experiments for identifying differentially expressed genes, the genes are naturally grouped into gene sets by their biological relevance/pathway information. 
It is often suspected, {\it a priori}, that genes in certain sets are more (or less) likely to be differentially expressed under the specific experimental condition.

Without loss of generality, we assume that testing data can be classified into $K$ mutually exclusive groups based on some known categorical annotation $\dv = (d_1,\dots, d_m)$, where $d_i$ is the corresponding group label for the data involved in the $i$-th test.   
Under this setting,  the exchangeability assumption in the original two-group model is no longer valid due to the availability of the covariates, $\dv$.
Nevertheless, it is straightforward to extend the CIHM from the two-group model into the following form to accommodate the non-exchangeable data:
\begin{equation}\label{multi.groups.model}
 \begin{aligned}
    &  \gamma_i \mid d_i  \sim {\rm Bernoulli} (1 - \pi_{d_i, 0}), \\
    &  z_i \mid \gamma_i = 0, \, d_i  \sim f_{d_i, 0}, \\
   &  z_i \mid \gamma_i = 1, \, d_i  \sim f_{d_i, 1}.
 \end{aligned}
\end{equation}

The Bayesian FDR control is coherent in dealing with the non-exchangeable data in multiple hypothesis testing. 
Specifically, the same decision rule (\ref{bayes.dec.rule}) is applied and the only notable difference is that the false discovery probability for the $i$-th test, $\Pr(\gamma_i = 1 \mid \zv, d_i)$  is now additionally conditional on the covariate data, $d_i$.
Under the oracle condition, i.e., as the quantities $\{ (\pi_{k,0}, f_{k,0}, f_{k,1}): k=1,...,K\}$ are known, the desired local fdr is computed by 
\begin{equation*}
  \Pr(\gamma_i = 1 \mid \zv, d_i) = \frac{\pi_{d_i,0} \, f_{d_i, 0}(z_i)}{\pi_{d_i,0} \, f_{d_i, 0}(z_i) + (1-\pi_{d_i,0}) \, f_{d_i,1}(z_i)},
\end{equation*}
which is monotonic to the weighted likelihood ratio (wlr) statistic, 
\begin{equation*}
 wlr^*(z_i;d_i) =\frac {1-\pi_{d_i,0}}{\pi_{d_i,0}} \, \frac{f_{d_i,1}(z_i)}{f_{d_i, 0}(z_i)}.
\end{equation*} 
Importantly, under the oracle condition for multiple testing with non-exhangeable data, \cite{Sun2007} shows that the Bayesian FDR control is optimal in the sense that it minimizes FNR while controlling the desired FDR level.
%In a general mixture model formulation, as an extension to the two-group model,  it translates to that genes in certain pre-defined gene sets may have lower (or higher) $\pi_0$ values than others. 
A natural question to ask is:   is there asymptotically equivalent frequentist procedure corresponding to the Bayesian procedure?

Let $F_{k, wlr^*}$ denote the CDF of the wlr statistic, $\frac {1-\pi_{k,0}}{\pi_{k,0}} \, \frac{f_{k,1}}{f_{k, 0}}$, for group $k$ under the null. 
We find the following procedure seemingly provides a frequentist FDR control equivalent. 
\begin{algorithm} {\rm (Frequentist FDR control of multiple hypothesis testing with non-exchangeable data)}
\begin{enumerate}
 \item Obtain the CDF of the weighted likelihood ratio statistics under the null hypothesis, $F_{k,wlr^*}$, for each group $k = 1, ..., K$.
 \item At a given threshold of the wlr statistics, namely t,  reject the set of tests $\{i: wlr^*_i \ge t\}$ across all groups. 
 \item Evaluate the frequentist FDR at threshold $t$ by 
  \begin{equation} \label{non-ex.ffdr}
                         \FDR(t) = \frac{\sum_{i=1}^m \pi_{d_i,0} [1-F_{d_i, wlr^*}(t)]}{ \left[\sum_{i=1}^m \lv_{\{wlr^*_i \ge t\}}\right] \vee 1}.
\end{equation}   
\end{enumerate}
\end{algorithm}
In practice, to apply Algorithm 1, we rank the wlr statistics from all tests (across groups) in {\em descending} order and evaluate $\FDR(t = wlr^*_{(i)})$ sequentially  until the pre-defined FDR control level is achieved.
We show, in Proposition 1, that Algorithm 1 controls frequentist FDR and is equivalent to its Bayesian counterpart as $m \to \infty$. 

\begin{prop} 
Under the oracle condition, the frequentist procedure described in Algorithm 1 is asymptotically equivalent to the Bayesian FDR control procedure for testing multiple hypotheses with non-exchangeable group structures, i.e., 
\begin{equation*}
    \BFDR\left ( t_{b,i} = u^*_{(i)} \right ) \xrightarrow{a.s.} \FDR \left (t_{f,i} = wlr_{(i)}^* \right) \mbox{ for all } i.
\end{equation*}
\end{prop}
\begin{proof}
  See Appendix C.
\end{proof}

Because of the asymptotic equivalence, we conclude that Algorithm 1 also shares the optimality of the Bayesian procedure under the setting. 
Although the quantity $1-F_{d_i, wlr^*}$ may be interpreted as a group-specific $p$-value, it is not directly involved in the decision rule as in the 
case for non-exchangeable data.
Furthermore, note that the ranking of the wlr statistics is generally different from the ranking of the group-specific $p$-values across groups.  
Many authors \citep{Genovese2006,IHW2016} also have explicitly sought the decision rules of the following form for testing with non-exchangeable data, 
\begin{equation}
   \delta_i(t) = \lv_{\{ \frac{p_i}{w_{d_i}} \le t\}}, 
\end{equation}
where $p_i$ represents the $p$-value derived from $\zv$ for the $i$-th test and $w_{d_i}$ represents a group-specific weight.
It should be noted that weighting $p$-values generally results in a very different rejection path than weighting the likelihood ratio statistics because of the irreconcilable difference in ranking individual tests and is therefore unlikely, if not impossible,  to achieve the optimal performance. 

In a practical setting when $\{ (\pi_{k,0}, f_{k,0}, f_{k,1}): k=1,...,K\}$ are unknown, the same approximate Bayesian inference (i.e., PEB) and model diagnosis strategies apply for the Bayesian FDR control procedure. 
On the other hand, extending the optimal oracle frequentist FDR control seems non-trivial: unlike in the two-group model setting, it is required to make explicit assumption and inference on the distributions from the alternative scenarios for each group $k$.
As a consequence, the operational procedure of an ideal frequentist FDR control method in this setting may inevitably resemble its Bayesian counterpart. 
 
Finally, we note that the parametric Bayesian approach can be consistently extended to account for continuous covariate data $\dv$. 
As a starting point, a CIHM can be specified as 
\begin{equation}
\begin{aligned}
   &  \gamma_i \mid d_i  \sim {\rm Bernoulli}\, \left(\frac{\exp [\alpha_0 + \alpha_1 d_i]}{1+\exp [\alpha_0 + \alpha_1 d_i]}\right), \\
   &  z_i \mid \gamma_i = 0 \, \sim f_0, \\
   &  z_i \mid \gamma_i = 1 \, \sim f_1.
\end{aligned}
\end{equation}
Note that the logistic prior model for $\gamma_i \mid d_i$ is applicable for the categorical covariate data:  it is simply an alternative parametrization of the $\gamma_i$ prior in (\ref{multi.groups.model}).   
Furthermore, it is also possible to specify parametric distributions of $f_0$ and $f_1$ that are dependent on the continuous covariate data. 
However, these modeling choices are largely context-dependent and outside the scope of this paper.

\section{Numerical illustrations}

\subsection{{\em locfdr}  vs. $q$-value}

The {\it locfdr} and the $q$-value procedures are two commonly applied procedures for Bayesian and frequentist FDR controls, respectively.  
The implementation of the {\it locfdr} method represents a frequentist inference alternative to estimate the Bayesian FDR, as it circumvents prior specification of $\pi_0$ and the explicit computation of parametric likelihood for alternatives. 
For the very reason, it is difficult to evaluate the rejection path of the local fdr method via Theorem 2. 
Here we conduct numerical experiments to compare the behaviors of the {\it locfdr} and the $q$-value methods.

To ensure both the {\it locfdr} and the $q$-value methods utilize the equivalent test statistics, we directly simulate $z^2$-statistics and obtain the $p$-values for the $q$-value method according to the $\chi^2$-distribution with 1 degree of freedom.
We then compute $z_i =  {\rm sgn}(z_i) \sqrt{z_i^2}$, for each test $i$, where ${\rm sgn}(z_i)$ is independently drawn from the set $\{+1, -1\}$ with the probability 0.5 for each choice.   
This transformation results in an overall symmetric $z$-score distribution, and the corresponding two-sided $p$-values derived from the $z$-statistics are identical to what is used in the $q$-value method.
We employ the two-groups model to simulate $z^2$ statistics with $\pi_0 = 0.55$ and $m = 20,000$. 
The null data are simulated from the theoretical $\chi^2_1$-distribution.
The $z^2$ statistics of the alternative models are generated from a family of $\Gamma(k, \theta)$ distributions, where $k$ and $\theta$ denote the shape and scale parameters, respectively.
We find that this family of distributions is convenient to describe a wide spectrum of alternative scenarios under both the UA and the ZA.  
In particular, if $z \sim {\rm N}(0, \phi^2)$, it follows that $z^2 \sim \Gamma (0.5, 2 \phi^2)$.  
Adjusting the shape parameter $k$ in the range of $(0, 0.5)$ results in heavy-tailed distributions of $z$-scores that resemble $t$ and double-exponential distributions.
In comparison, setting $k \gg 0.5$ leads to the distribution of the simulated $z$-scores with two modes that are symmetric to 0, and the larger $k$ value results in greater separation of the two modes. 
Under this setting, the ZA seems reasonable for $k > 0.5$, otherwise, the generated data are better described by the UA assumption.
For inference,  we apply the same non-parametric estimator of $\pi_0$ (\ref{q.est})  for both the {\it locfdr} and the $q$-value methods by setting $\eta = 0.5$ and assuming the true theoretical null distribution of $z$-scores, i.e., ${\rm N}(0,1)$.

The simulation results show that when $k \ge 0.5$, the two methods always yield highly concordant rejection paths, and the practical difference is typically negligible. 
In comparison, as $k$ is decreased from 0.5 toward 0,  the rejection paths become noticeably different: although both approaches behave conservatively,  the {\it locfdr} method is significantly more conservative than the $q$-value method. 
In light of Theorem 2, the observation indicates that the ranking of the data becomes increasingly different by the two approaches as $k$ decreases from  0.5 towards 0. 
We examine the empirical rank correlations of $p$-values and the values of the inferred local fdr's from the simulated data with respect to a range of $k$ values. 
Figure \ref{rank_corr.fig} clearly verifies the observed pattern. 
This can be explained by the fact that likelihood ratio statistics $f_1(|z|)/f_0(|z|)$ may not be monotonically decreasing with respect to $|z|$ for heavy-tailed distributions.

The most important finding from this exercise is that comparing the rejection paths (or the rankings of test statistics) between the {\em locfdr} and the $q$-value method may be helpful in diagnosing the zero assumption, as the concordance of the rejections seemingly suggests the validity of the ZA. 
It is worth noting that in most cases the divergence between the rejection paths becomes noticeable only when the pre-defined control level is large. 
For practically applied stringent FDR control levels (e.g., $ \alpha = 0.05$), the two approaches generally agree well in all situations. 
(This is because the ranking of data points at extreme tails are typically consistent in both $p$-values and likelihood ratio statistics.)  
Nevertheless, the complete rejection path, rather than the part under some stringent threshold, is informative for validating the model assumptions.

\begin{figure}
\centering
\begin{tabular}{c c}
    \includegraphics[width=.40\textwidth]{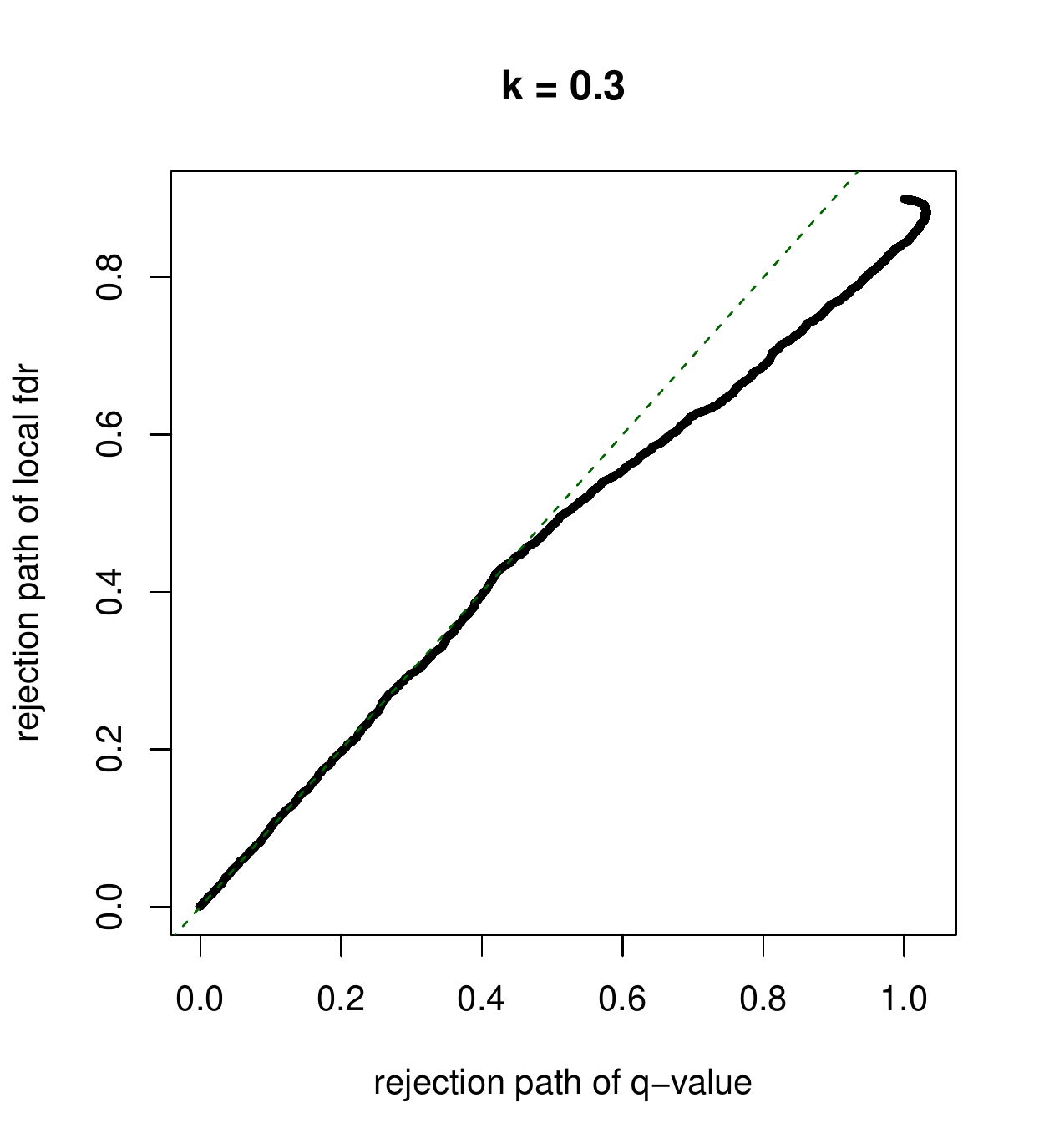}  &        \includegraphics[width=.40\textwidth]{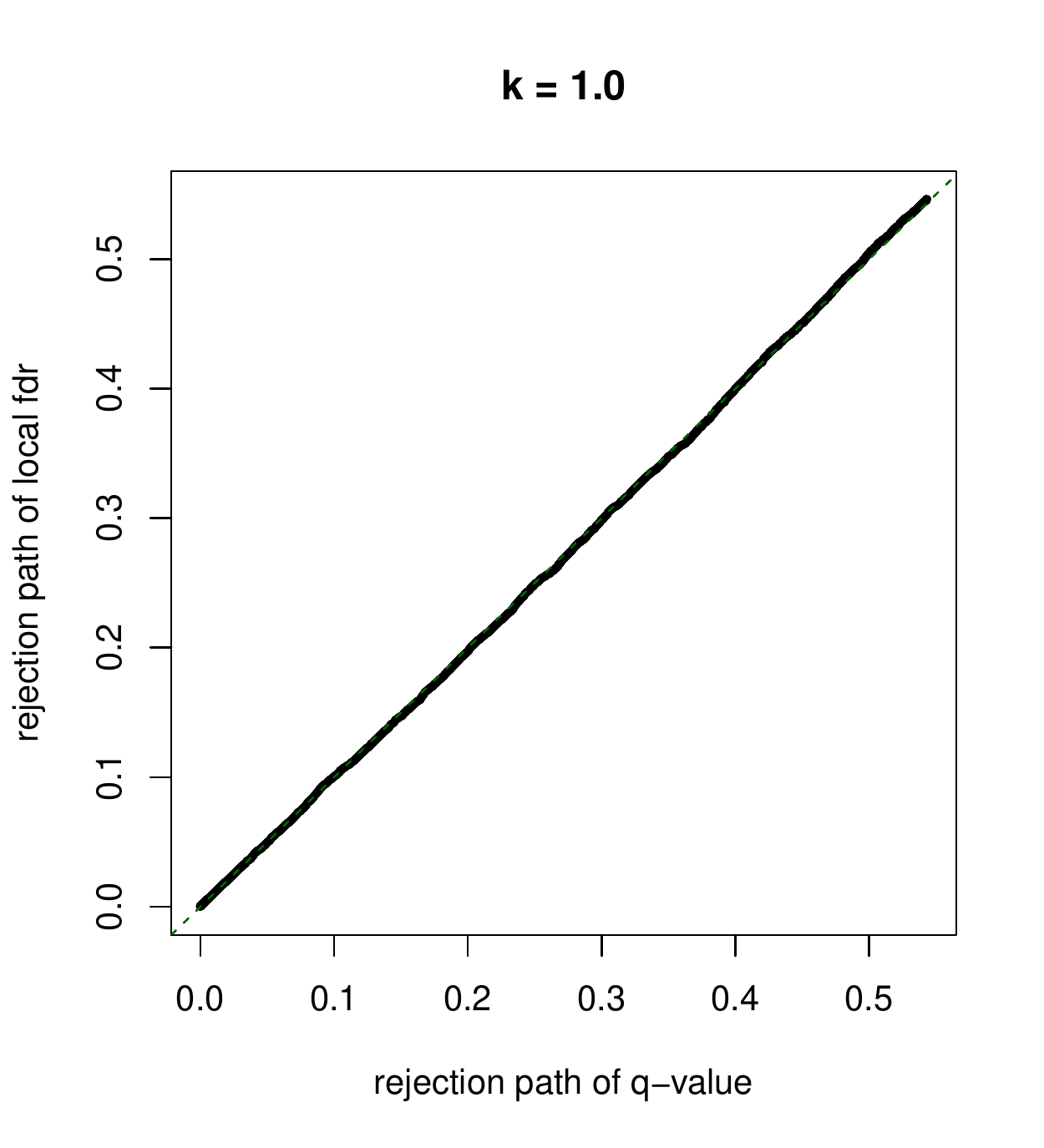} \\
 \multicolumn{2}{c}{\includegraphics[width=.80\textwidth]{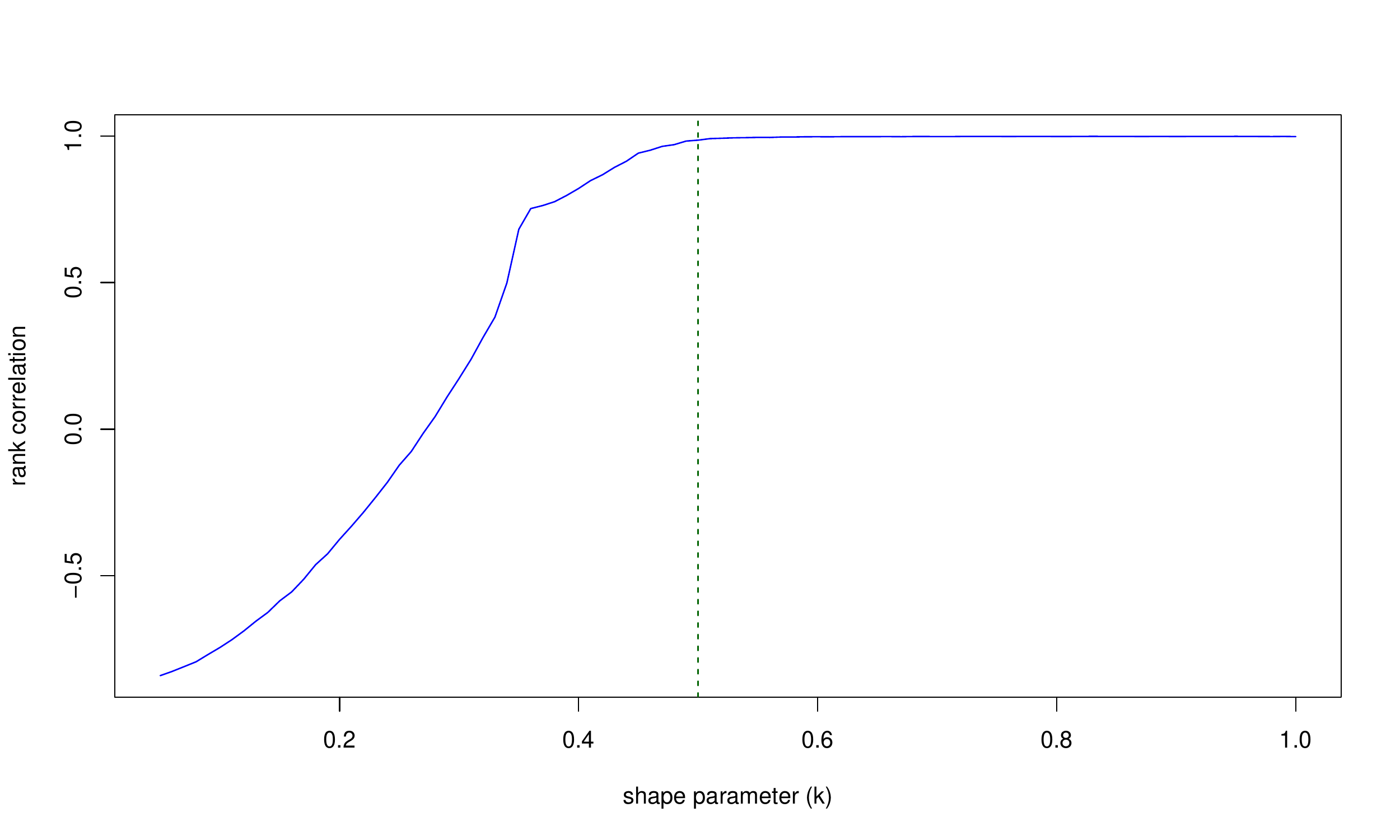} } \\
 \end{tabular}     
\caption{{\bf Numerical comparison between the {\it locfdr} and the $q$-value methods.} The top-left panel shows the disconcordance of the rejection paths between the two approaches as the ZA is severely violated (shape parameter set to 0.3 for the alternative distribution). The top-right panel indicates the concorddance of the two rejection paths when the ZA is reasonable (shape parameter $ = 1.0$ for the alternative distribution). The bottom panel shows the rank correlation between $p$-values (used by the $q$-value method) and the local fdr's computed by the {\it locfdr} method. Clearly, as the ZA is severely violated, the rank correlation significantly deviates from 1. Based on Theorem 2, this explaines the disconcordance of the rejection paths between the {\em locfdr} and the $q$-value methods \label{rank_corr.fig}}
\end{figure}

Finally, we apply the comparison to the Hedenfalk data distributed in the $q$-value package. The data set includes 3,220 pre-computed $p$-values from a differential gene expression study.  We transform the provided $p$-values to the corresponding $z$-scores according to the distribution function of the standard normal and apply both the {\it locfdr} and the $q$-value methods using the same $\hat \pi_0$ estimation method described above. 
The comparison of the rejection paths indicates the results of the two approaches have an excellent agreement (Figure \ref{lfdr_vs_q_hedenfalk.fig}).

\begin{figure}
\centering
 \includegraphics[width=.50\textwidth]{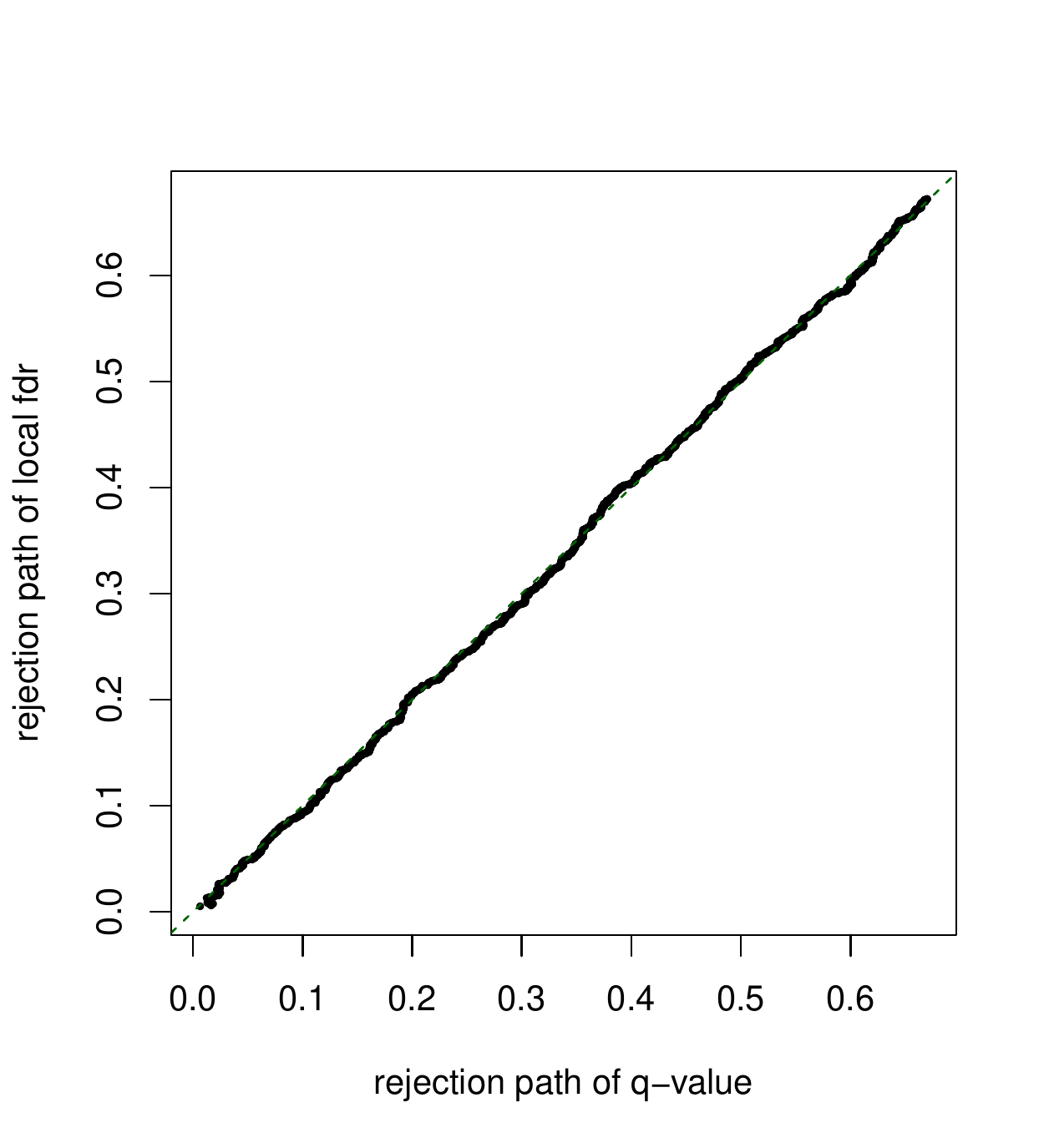} 
\caption{{\bf comparison of rejection paths of the {\it lfdr} and the $q$-value methods in Hedenfalk data.} Although the number of tests is modest ($m = 3170$), the two rejection paths appear to be highly concordant.\label{lfdr_vs_q_hedenfalk.fig} }
\end{figure}

\subsection{Model Diagnosis for Bayesian FDR Control Procedures}

We perform simulations to illustrate the importance of post-fitting model diagnosis for Bayesian FDR control procedures, especially for the PEB-based approaches. 
Adopting a similar simulation scheme described in the previous section (i.e., simulating $z$ or $z^2$-statistics from various alternative distributions including heavy-tailed, bi-modal distributions at the two different extremes),  we generate data from the two-group model with a spectrum of alternative scenarios. 
Our aim is to examine the robustness of the PEB-based Bayesian FDR control procedures under different modes of model misspecification.
We are also interested in investigating the ability of standard model diagnosis techniques for detecting the inadequate model fitting and assessing its consequence on FDR control. 

For demonstration, we focus on a recently proposed Bayesian FDR control method, ASH, a PEB implementation assuming the UA \citep{Stephens2016}.
Briefly, ASH, implemented in the R package {\em ashr}, has the ability to flexibly model uni-modal alternative distributions from a mixture of a rich class of base distributions. 
The weights for each basis distribution function are estimated from data by an EM algorithm. 
To serve our purpose in this demonstration, we intentionally limit the base distributions to the class of Normal distributions with mean 0 and variable values of variance parameters, i.e., the alternative model assumes
\begin{equation*}
    f_1 = \sum_{k=1}^K  \omega_k \, {\rm N}(0, \sigma_k^2),
\end{equation*}
where $\omega_k$'s are the weight parameters to be estimated by the EM algorithm and  a grid  of $\sigma_k$'s values are pre-selected in a data-driven way \citep{Stephens2016}.
It is also important to note that the Bayes factors computed under the assumed alternative model are monotonic to the two-sided $p$-values derived from the $z$-statistics.
We anticipate that this modeling strategy would be robust (although not always perfect) if the UA assumption is indeed correct, but potentially problematic if the ZA actually holds.

To perform model diagnosis, we compare the fitted mixture distribution, 
\begin{equation*}
 \hat f_c = \hat \pi_0 f_0 + (1 - \hat \pi_0) \sum_{k=1}^K  \hat \omega_k \, {\rm N}(0, \sigma_k^2) 
\end{equation*}
with the observed mixture distribution from data.
In particular, we compare a full range of theoretical quantiles computed from $\hat f_c$ with the observed sample quantiles. 
Because of the large sample size employed in the simulation ($m = 10,000$), we are able to assess an asymptotic $p$-value for each pair of quantile values compared \citep{FergusonBook}, which we use to evaluate the severity of the mismatch between the fitted theoretical quantiles and the observed ones. 
This strategy is indeed a standard approach to Bayesian model diagnosis by noting that $\hat f_c$ is an approximation of posterior predictive distribution as shown in Equation (\ref{ppd.appx}).

For each simulated data set, we also construct an oracle rejection path by computing the likelihood ratios/Bayes factors assuming the true alternative model and plugging in the true $\pi_0$.
We also similarly construct an {\em expected} rejection path by utilizing a two-sided $p$-value for each test data point and the estimated value of $\hat \pi_0$ from the EM algorithm. 
Comparing the rejection path from the Bayesian approach to the expected rejection path should provide some assessment of goodness-of-fit according to the theoretical result of Theorem 2.  

First, we highlight the results from three distinct alternative distributions.
\paragraph{1. Scaled $t$-distribution with 10 degrees of freedom.} This alternative distribution follows the UA but modestly deviates from the specific parametric assumption.   
Despite the imperfect model specification, the model diagnosis indicates that the fit by ASH is mostly adequate (Table \ref{sim_t.tbl}), i.e., there are no obvious mismatches between the estimated and sample quantiles. 
The comparison between the ASH rejection path and the oracle path (Figure \ref{sim_compare.fig}) indicates that the ASH result is slightly conservative, i.e., FDRs are over-estimated by ASH, especially at more relaxed threshold values.
The $q$-value method behave much more conservatively, judging by the comparison of its rejection path with the oracle path.     
Finally, we find that the ASH rejection path is highly concordant to the expected rejection path, which validates the result of model diagnosis and Theorem 2. 

\begin{table}[h!]
\centering
\caption{\label{sim_t.tbl} Model diagnosis for alternative scenario 1 ($t$-distribution) } 
 \begin{tabular}{c | c c c c c c c c c c } 
 \hline
  ~  & 10\% & 20\% & 30\% & 40\% & 50\% & 60\% & 70\% & 80\% & 90\% \\ 
 \hline
 sample quantile &  0.051 & 0.214 & 0.495 & 0.962 & 1.710 & 2.955 & 5.116 & 8.811 & 16.81 \\
 fitted quantile  &  0.052 & 0.214 & 0.509 & 0.983 & 1.728 & 2.928 & 4.974 & 8.723 & 16.76 \\
 \hline
 $p$-value & 0.781 & 0.935 & 0.415 & 0.446 & 0.700 & 0.723 & 0.249 & 0.667 & 0.910\\
 \hline 
 \end{tabular}

\end{table}

\paragraph{2. Scaled double-exponential distribution.} This alternative distribution follows the UA but has a much heavier tail than the assumed alternative model could sufficiently capture.
The model diagnosis procedure can effectively detect the inadequate fitting in this case. 
Importantly, we identify a  pattern in all severe mismatches between the estimated and sample quantiles: the theoretical quantiles of the fitted distribution are consistently {\em over-estimated} (Table \ref{sim_doublex.tbl}). 
Furthermore, there is also noticeable discordance between the ASH rejection path and the expected rejection path, indicating that the goodness-of-fit condition of Theorem 2 is violated (Figure \ref{sim_compare.fig}).
Comparing to the oracle rejection path, ASH, local fdr and $q$-value approaches all display significantly conservative behaviors, which can be explained by the fact that all three methods over-estimate $\pi_0$. 
However, ASH exhibits a less degree of conservativeness and much-improved power over the existing the ZA-based approaches. 

\begin{table}[h!]
\centering
\caption{\label{sim_doublex.tbl} Model diagnosis for alternative scenario 2 (double exponential distribution)} 
 \begin{tabular}{c | c c c c c c c c c c } 
 \hline
  ~  & 10\% & 20\% & 30\% & 40\% & 50\% & 60\% & 70\% & 80\% & 90\% \\ 
 \hline
 sample quantile &  0.024 & 0.103 & 0.238 & 0.461 & 0.826 & 1.400 & 2.328 & 4.175 & 9.144 \\

 fitted quantile  & 0.028 & 0.114 & 0.266 & 0.503 & 0.859 & 1.403 & 2.305 & 4.093  & 9.331\\ 
 \hline
 $p$-value & $0.048^*$ & $0.021^*$ &  $0.001^*$ & $0.003^*$ & 0.129 & 0.938 & 0.672 & 0.447  & 0.523 \\
 \hline 
 \end{tabular}

\end{table}

\paragraph{3. Bimodal distribution induced by $\Gamma$ distribution with $ k = 0.7$} This alternative distribution departs from the UA  and is more properly described by the ZA.
In this scenario, the model diagnosis procedure also identifies the severely poor fit of the data. 
Importantly, we note that the pattern of mismatch between the estimated and observed quantiles are very different from the case of heavy tail alternative distributions:  the theoretical quantiles of the fitted distribution are consistently {\em under-estimated} (Table \ref{sim_bim.tbl}).  
Furthermore, there is also noticeable discordance between the ASH rejection path and the expected rejection path, indicating that the goodness-of-fit condition in Theorem 2 is violated.
In comparison to the oracle rejection path, we note that ASH is severely under-estimate $\pi_0$, which leads to dangerous anti-conservative behavior in FDR control, whereas the $q$-value method behaves properly under this scenario (Figure \ref{sim_compare.fig}).

\begin{table}[h!]
\centering
 \caption{\label{sim_bim.tbl} Model diagnosis for alternative scenario 3 (bimodal distribution) }
 \begin{tabular}{c | c c c c c c c c c c } 
 \hline
  ~  & 10\% & 20\% & 30\% & 40\% & 50\% & 60\% & 70\% & 80\% & 90\% \\ 
 \hline
 sample quantile &  0.099 & 0.371 & 0.823 & 1.481 & 2.418 & 3.774 & 5.827 & 9.023 & 14.66\\

 fitted quantile & 0.079 & 0.321  & 0.750 & 1.405 &  2.357 & 3.728 & 5.747 & 8.929 & 14.94\\ 
 \hline
 $p$-value & $<0.001^*$ & $<0.001^*$ & $0.003^*$&  $0.048^*$ & 0.281 & 0.570 & 0.480 & 0.570 &  0.279 \\
 \hline 
 \end{tabular}

\end{table}

\begin{figure}
\centering
\begin{tabular}{c c c}
    \includegraphics[width=.25\textwidth]{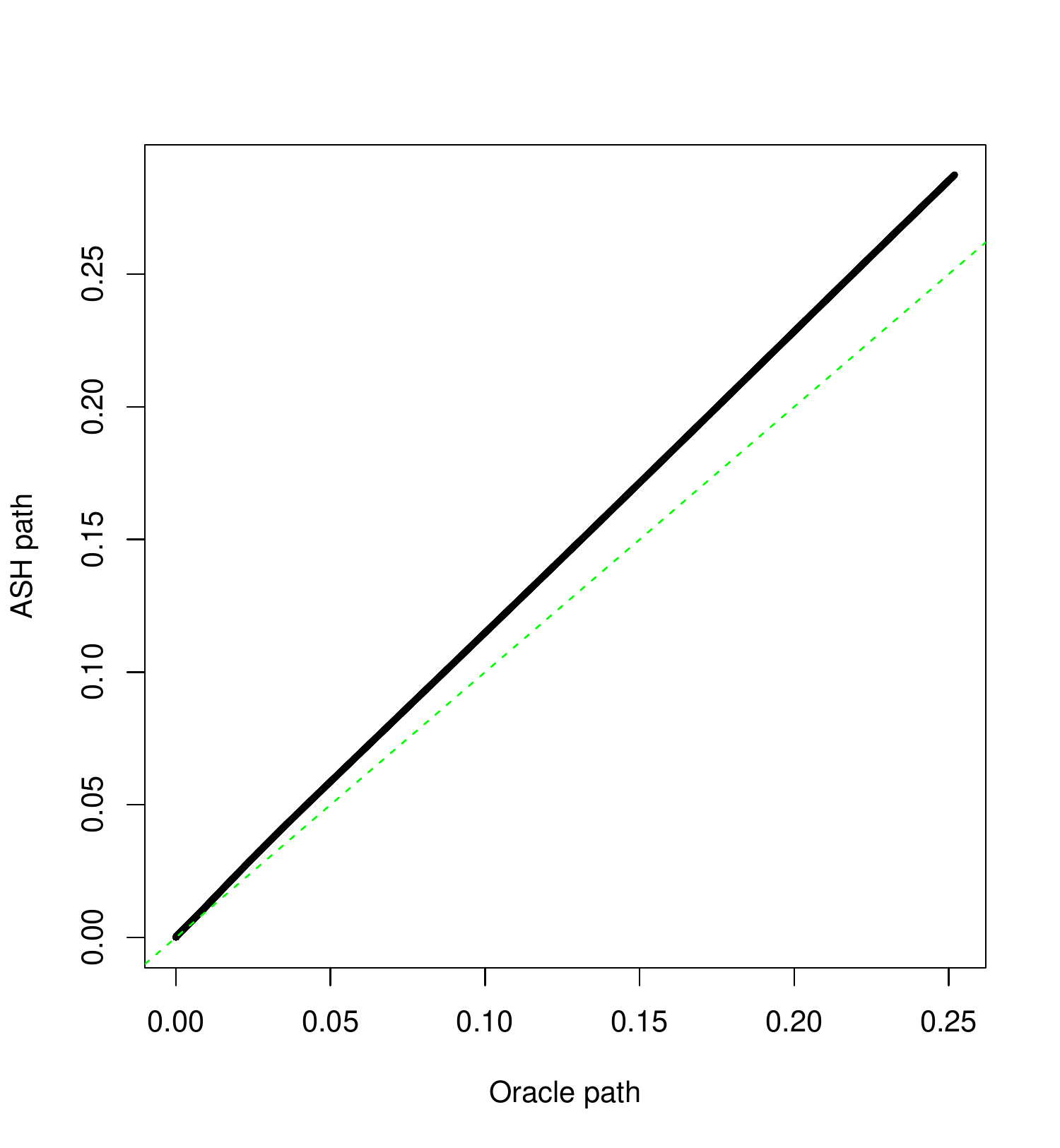}  &    \includegraphics[width=.25\textwidth]{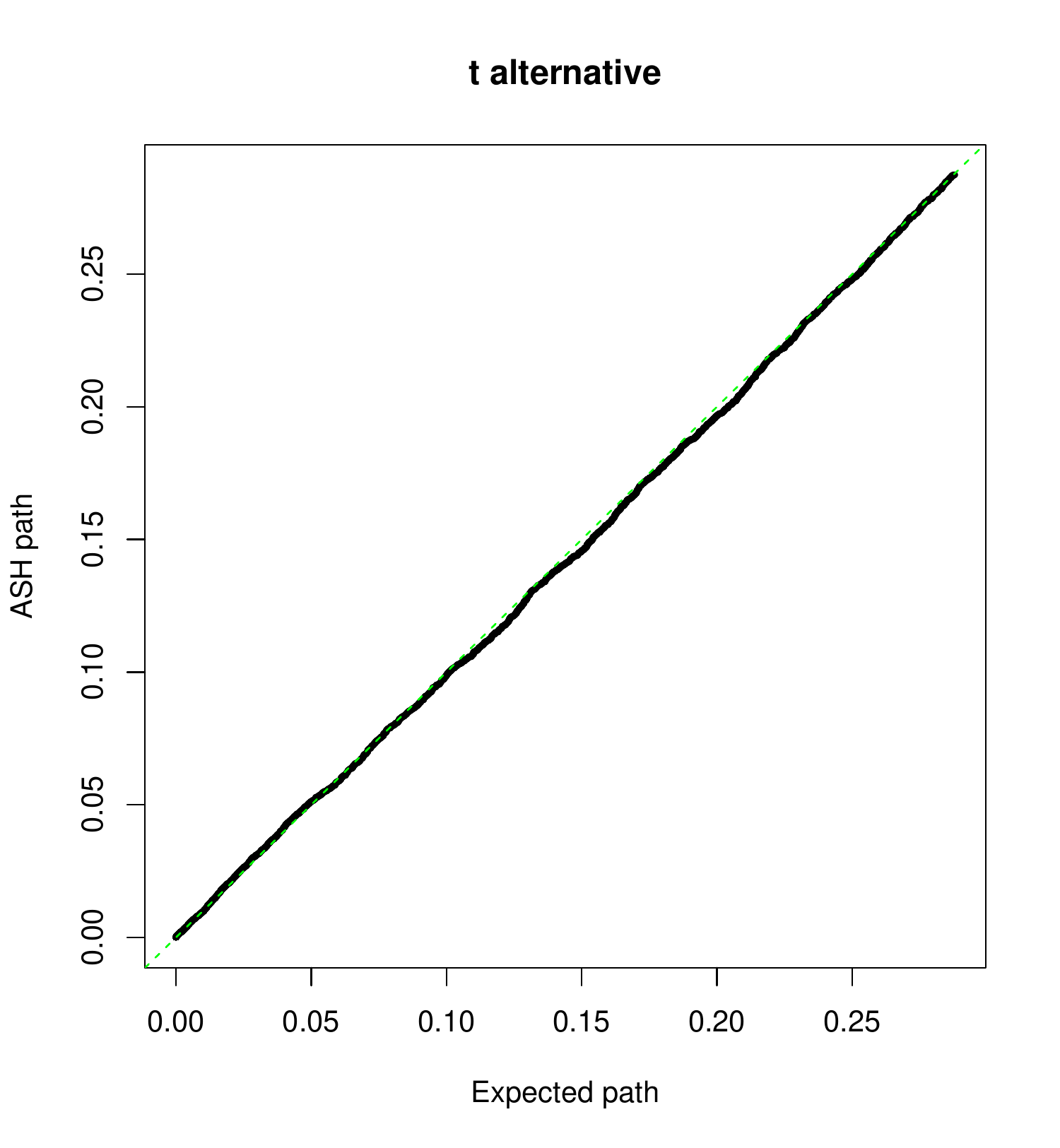} & \includegraphics[width=.25\textwidth]{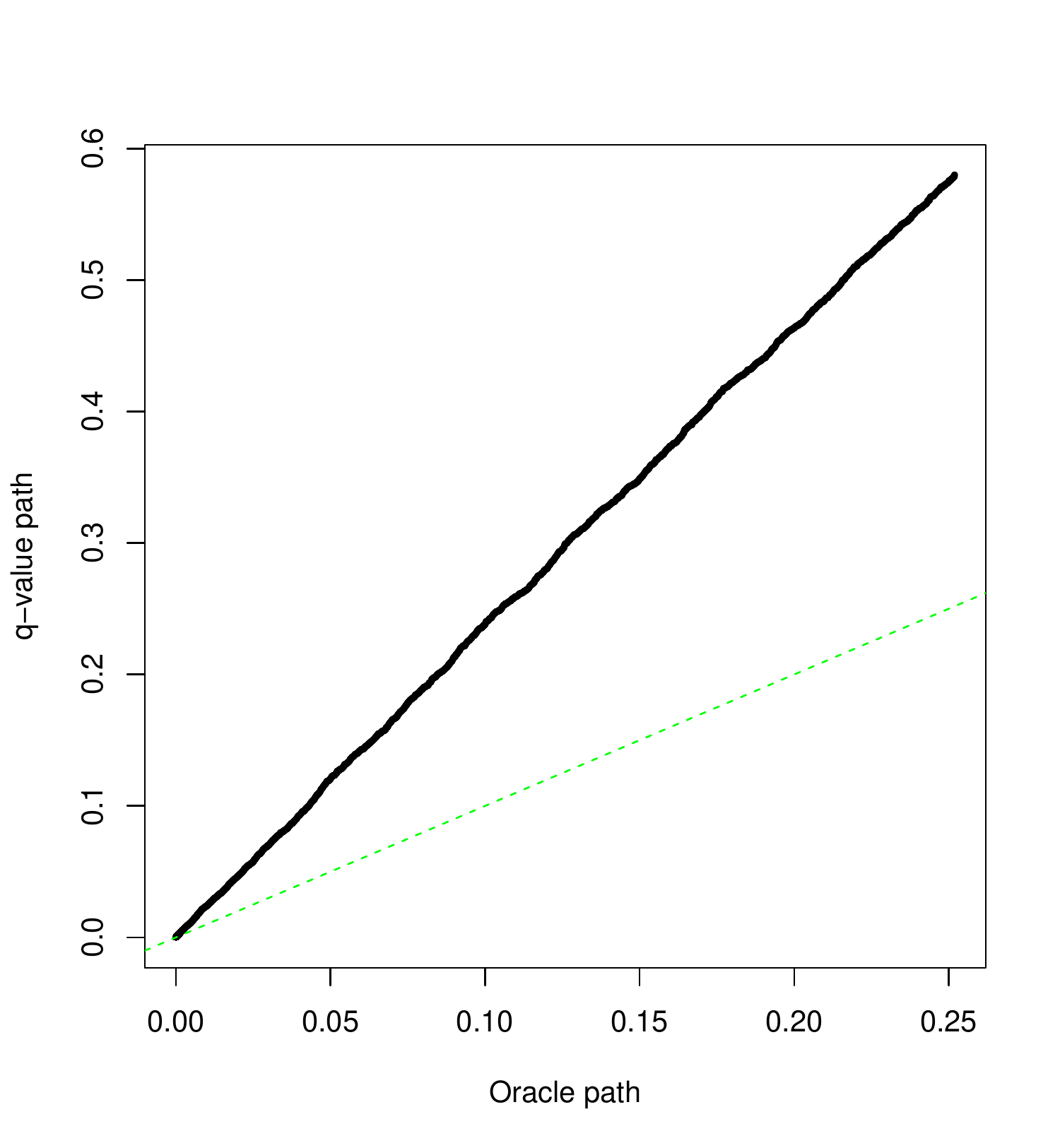}  \\
   \includegraphics[width=.25\textwidth]{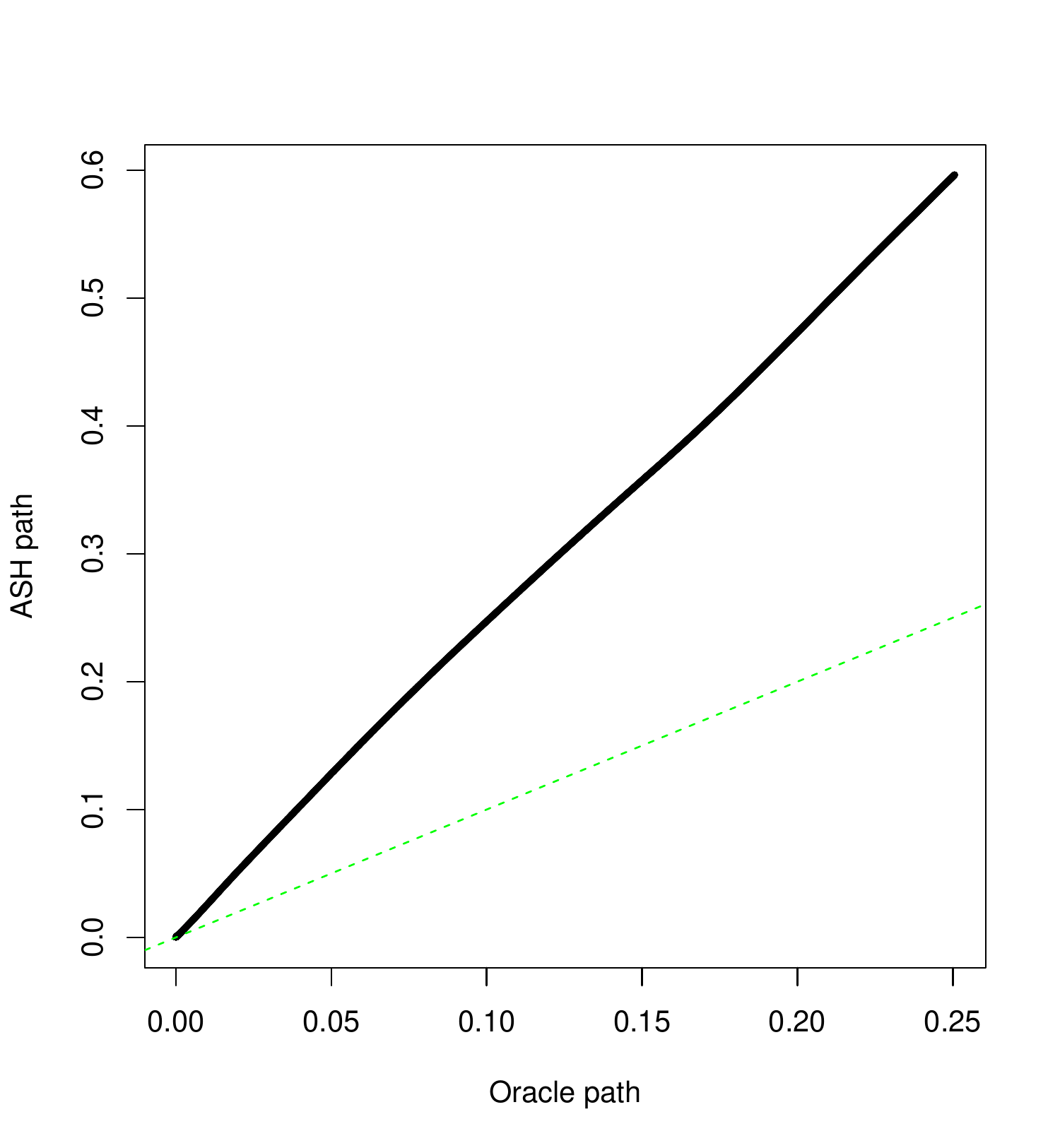}  &    \includegraphics[width=.25\textwidth]{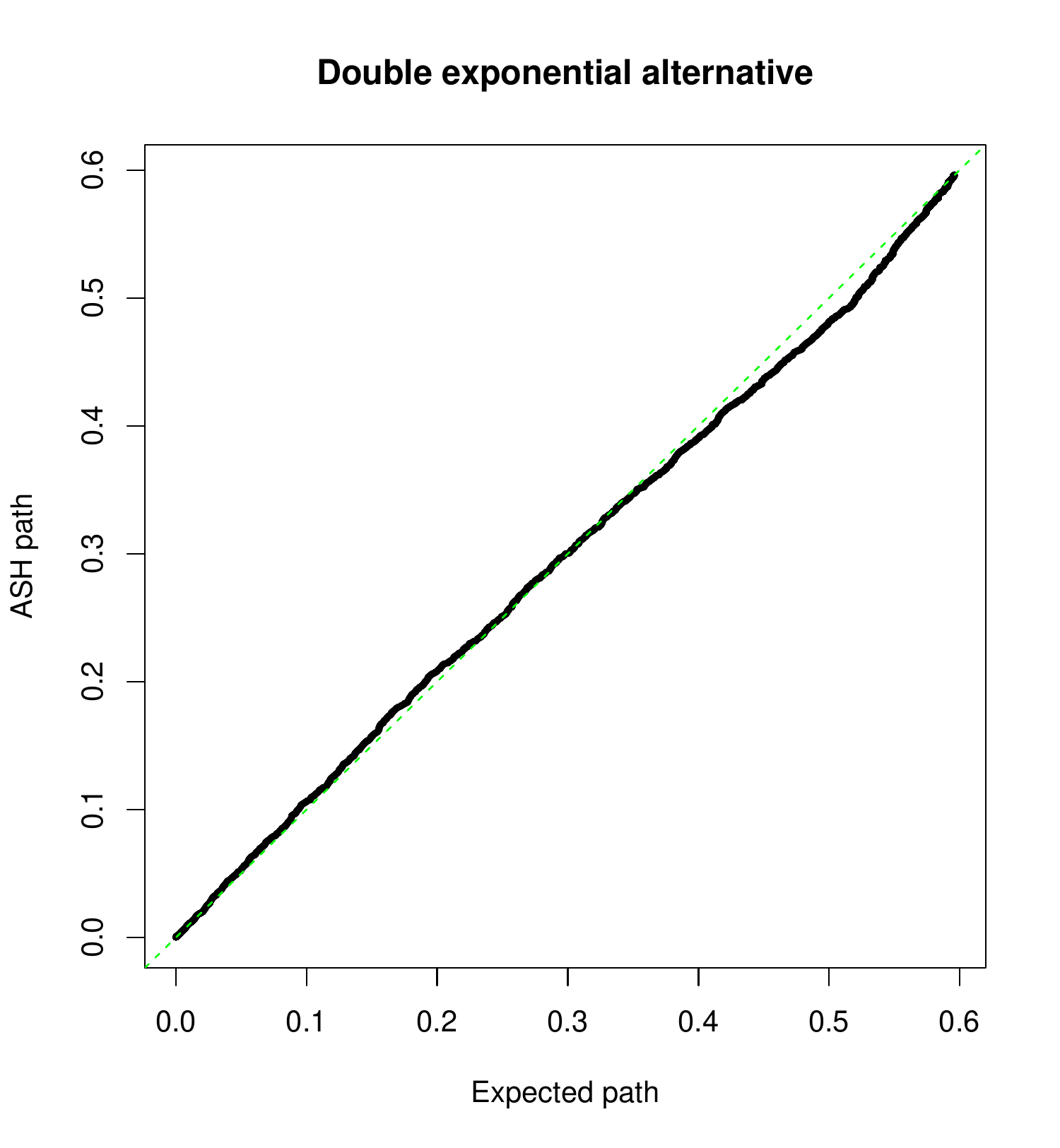} & \includegraphics[width=.25\textwidth]{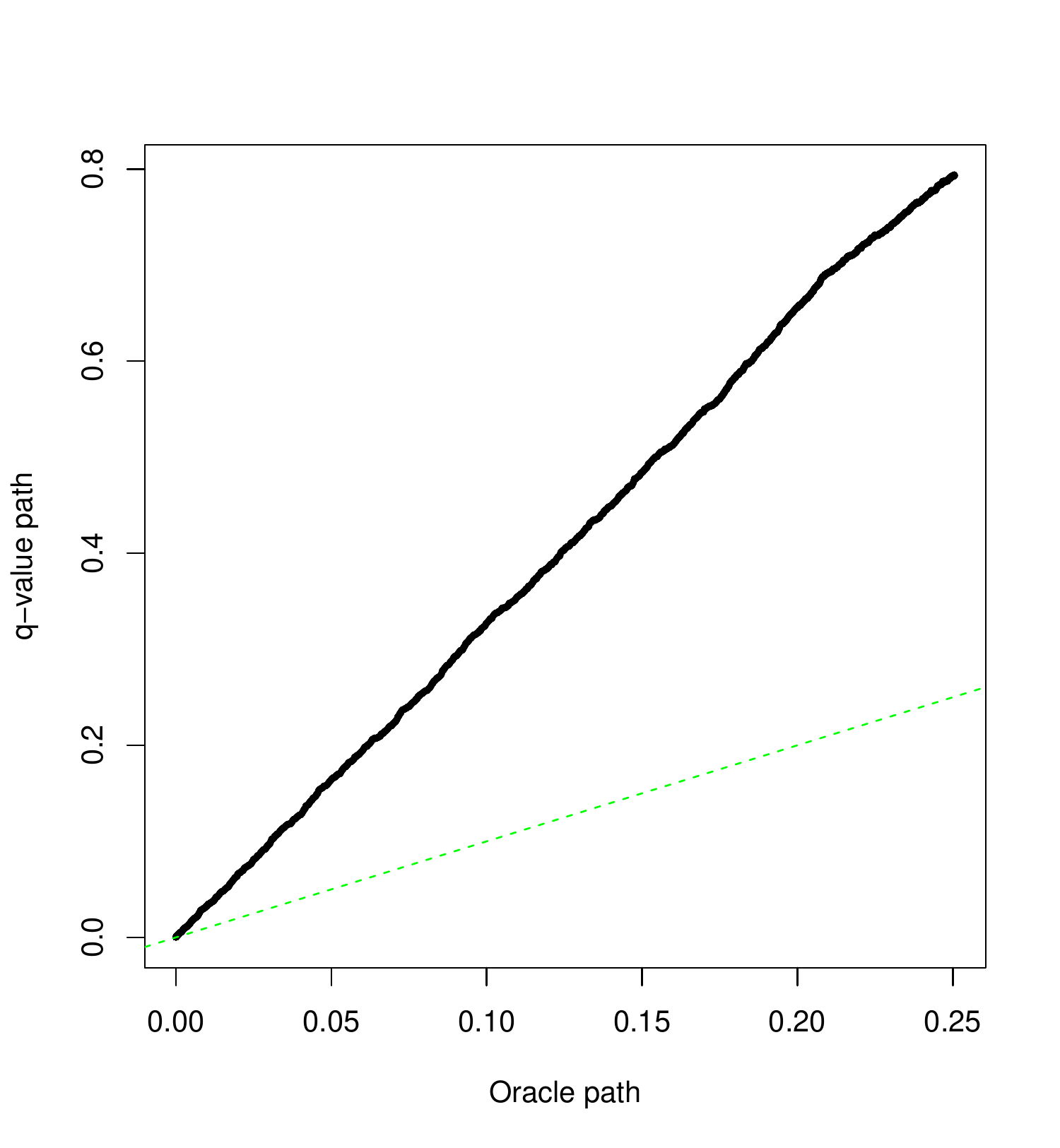}  \\
   \includegraphics[width=.25\textwidth]{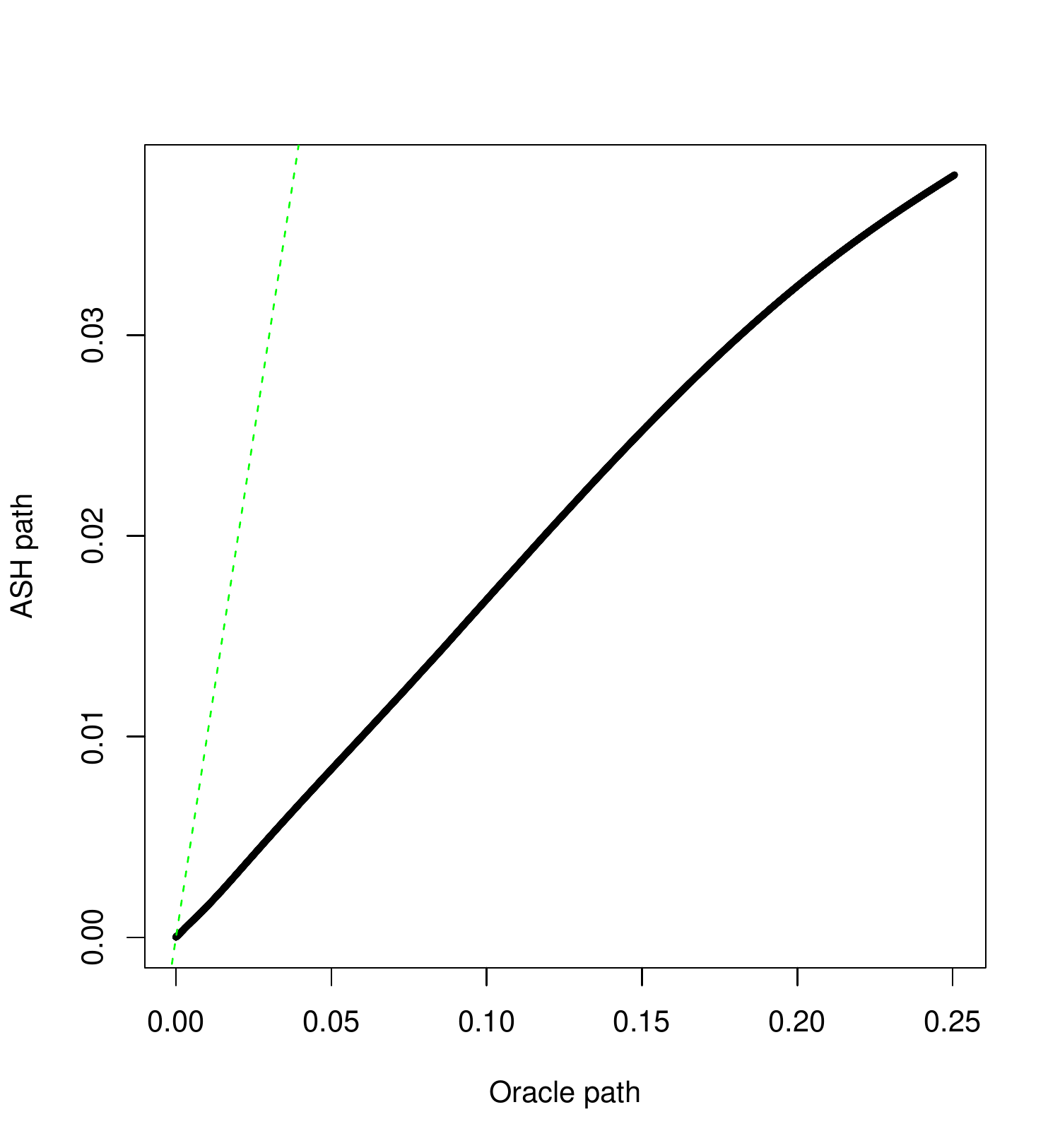}  &    \includegraphics[width=.25\textwidth]{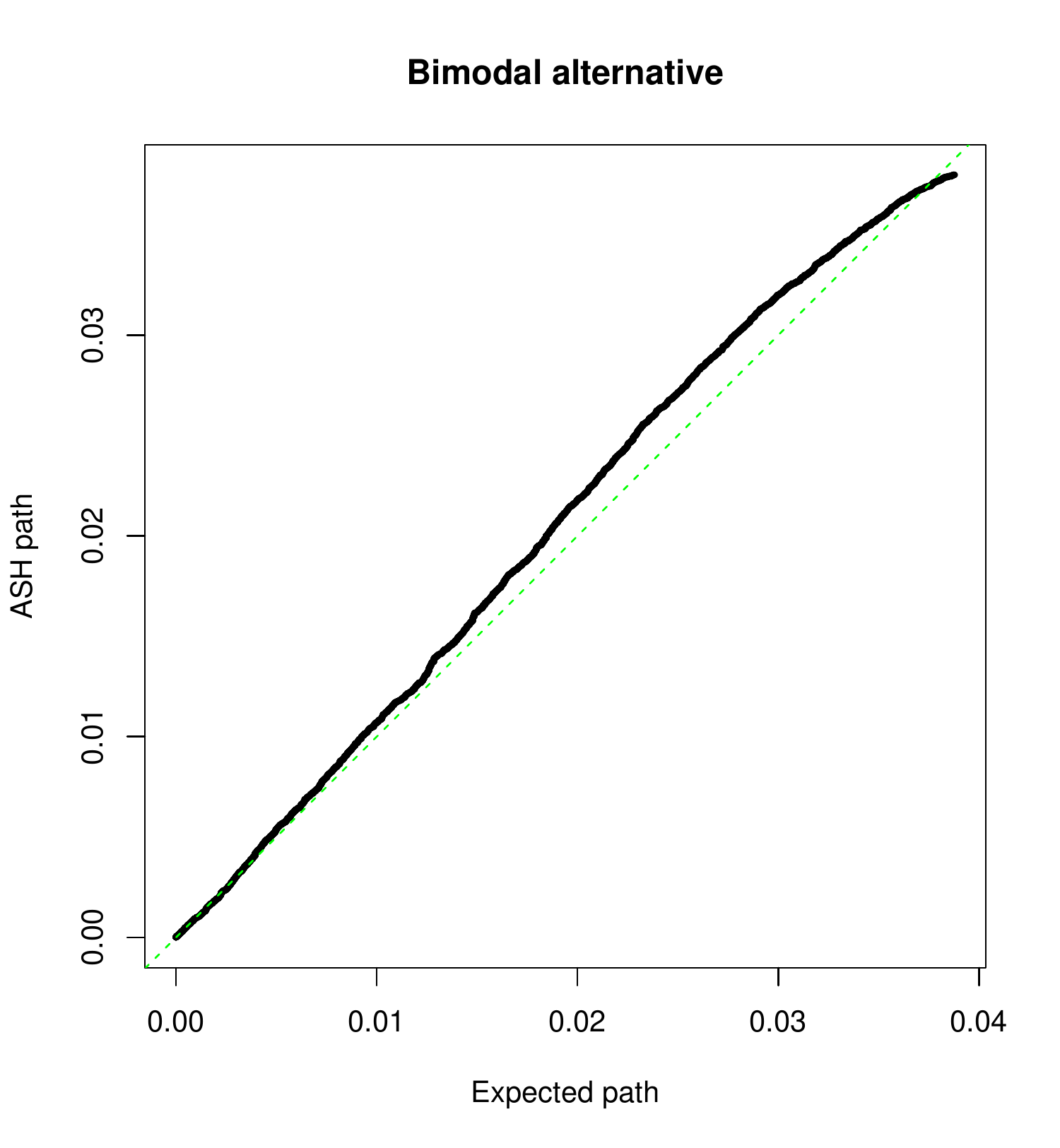} & \includegraphics[width=.25\textwidth]{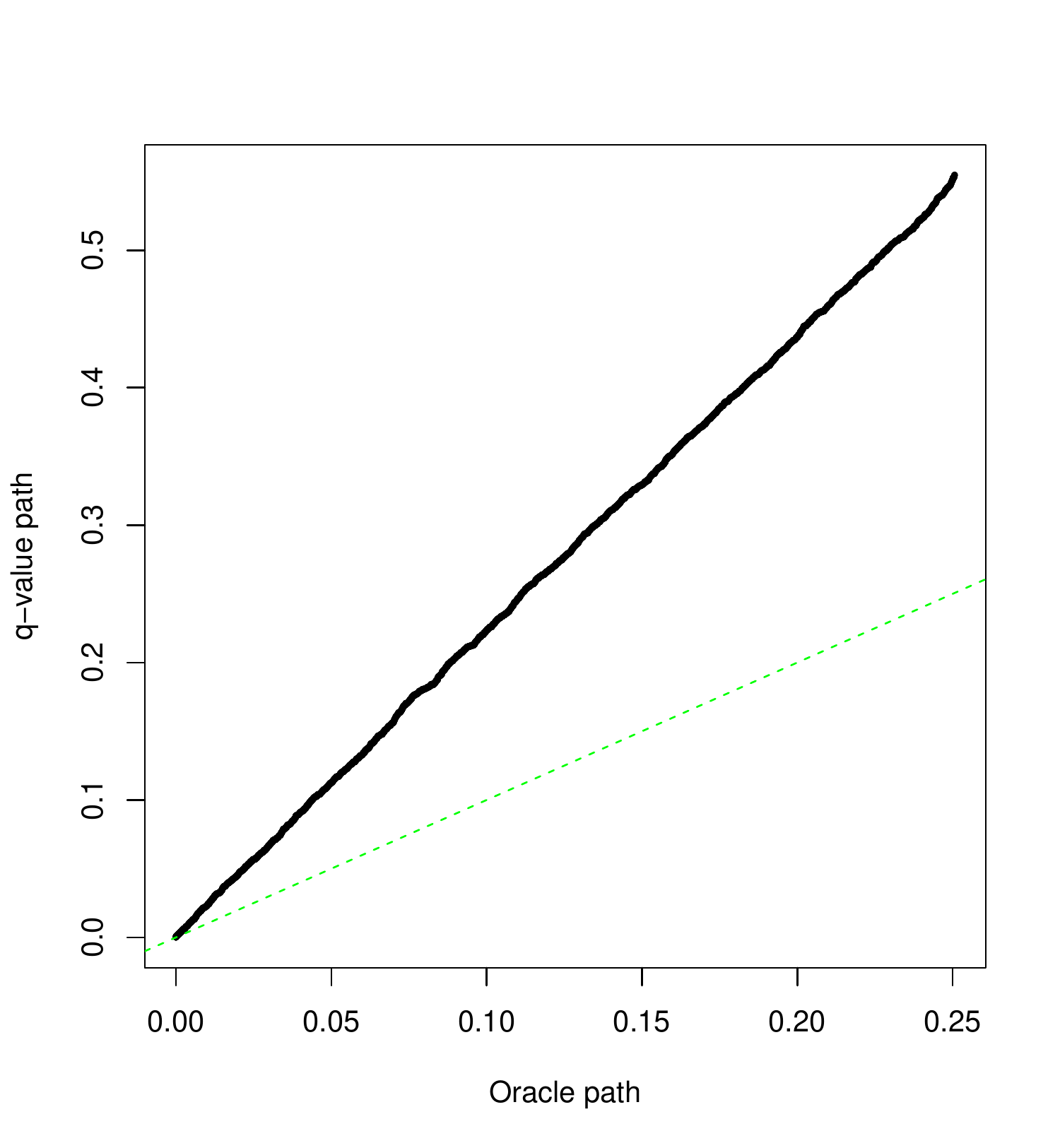}  \\
 \end{tabular}     
\caption{\label{sim_compare.fig}{\bf evaluation and diagnosis of PEB-based Bayesian FDR control procedure assuming UA.} The inference is conducted by the ASH method.  Each row represents a unique alternative distribution: the $t$ and double exponential distributions satisfy the UA (but not the particular parametric assumptions by ASH), the bimodal alternative severely violates the UA assumption. The first column shows the comparison of the observed rejection paths by ASH and the oracle rejection paths for each alternative distribution; the second column compares the observed rejection paths by ASH and the expected rejection paths; and the third column provides the reference comparison between the rejection paths of the $q$-value method and the oracle rejection paths. }
\end{figure}
 
Summarizing from these observations, we conclude that i) model diagnosis method is essential and effective in validating Bayesian FDR control methods; ii) even the data may be inadequately fitted by the assumed parametric model, identification of specific patterns of misfitting can be helpful to determine the conservativeness/anti-conservativeness of a method in FDR control.
It is worth emphasizing that specific patterns of misfitting typically depends on the inference methods/parametric assumptions. 
In the case of ASH, the specific mismatch patterns can be reasonably explained by the following observation: 
with flexible base distributions, ASH typically provides very good fit at the tails of the mixture distributions, where the data points have higher influence on the overall (marginal) likelihood; this behavior, nevertheless, causes inadequate fit for data of the non-tail areas and the specific patterns can be detected by examining a full range of quantiles.

To verify this finding, we conduct additional simulations of $z$-scores using the alternative models from a family of $\Gamma$ distributions with the shape parameter varying from 0.1 to 0.9. For each shape parameter, we generate 100 data sets with $m = 10,000$.
We apply both ASH and $q$-value methods to compute the rejection paths and perform the described model diagnosis procedures. 
Figure  \ref{pi0_est.fig} highlights the difference of $\pi_0$ estimation by the UA (ASH) and the ZA ($q$-value) assumptions. Under the ZA, the $\hat \pi_0$ estimates remain upper-bounds of the true value, but noticeably more conservative when the UA assumption is true. 
On the other hand, we confirm that the UA assumption can lead to severe under-estimation of $\pi_0$ when the ZA is indeed appropriate. 
Finally, the simulation results indicate that the simple model diagnosis method has good sensitivity and specificity in detecting the anti-conservative FDR control behaviors by ASH.
In particular, we flag a test case if at least one theoretical quantile under-estimates the corresponding sample quantile and the mismatch exceeds a pre-defined $p$-value threshold. 
For a suggestive threshold of  $p$-value $= 0.05$, the model diagnosis procedure flags 70.1\% of the test cases where the UA assumption is violated (i.e., the shape parameter $\ge 0.6$). In comparison, it only flags 4.2\% of the test cases when UA assumption is correctly specified (i.e., the shape parameters $\le 0.5$)

\begin{figure}
\centering
\includegraphics[width=.75\textwidth]{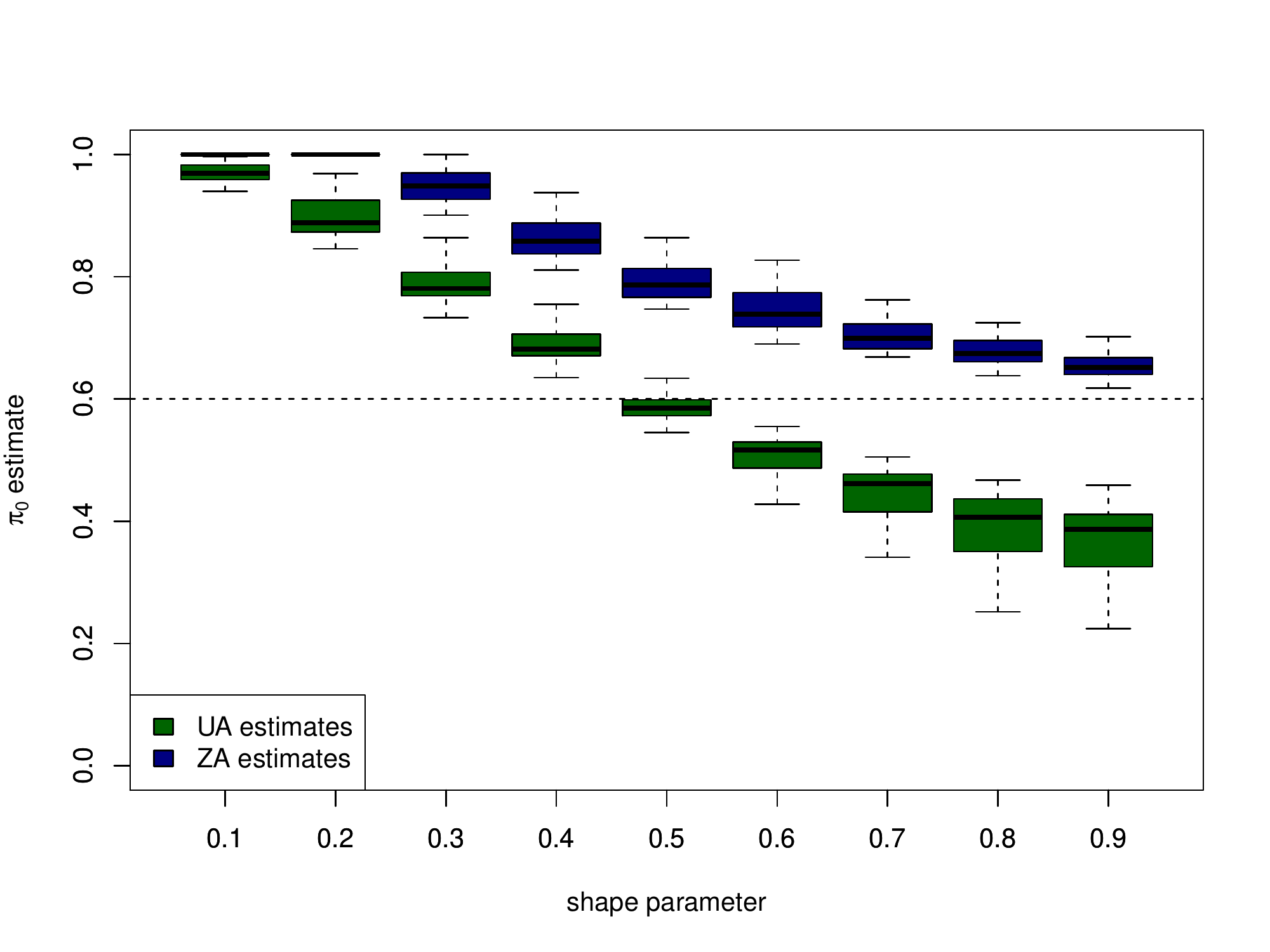}
\caption{\label{pi0_est.fig} {\bf Comparison of UA and ZA estimations of $\pi_0$ in different settings.} The true $\pi_0$, indicated by the dotted horizontal line, is 0.60.  The ZA estimates are obtained by the nonparametric quantile estimator (\ref{q.est}), and the UA estimates are obtained by the EM algorithm assuming a mixture normal alternative distribution.
The ZA estimates are uniformly conservative in all scenarios. The UA estimates, in comparison, are more accurate when the assumption holds (i.e., shape parameter $\le 0.5$); but can be significantly under-estimated if the assumption is severely violated (i.e., shape parameter $ > 0.5$.  }
\end{figure}

\section{Conclusion and discussion}

Our results described in this paper provide a theoretical connection between Bayesian and frequentist FDR control methods in the presence of a large number of tests. 
The setting of large-scale multiple testing creates a near asymptotic scenario where the frequentist property of the Bayesian method can be practically examined.
Our results also have profound implications for applying parametric Bayesian FDR control methods in a practical setting, in particular, the critical importance of post-fitting model diagnosis.
From our theoretical arguments and numerical experiments, we conclude that thorough and careful model checking can indeed improve the robustness of the parametric Bayesian FDR control procedures.
More specifically, we have showcased the numerical recipes in examining both the ZA (by examning the concordance of rejection paths between the {\it locfdr} and the $q$-value methods) and the UA (by Bayesian model diagnosis) assumptions.
Beyond the settings of the two-groups model, we argue that the Bayesian approaches show superior advantages in extensibility and flexibility over the $p$-value based frequentist approaches. 
In fact, we have shown that the ideal/optimal frequentist procedure under such setting should very much resemble the corresponding Bayesian approach.

It has been previously shown in the literature (e.g., \cite{Sun2007}) that Bayesian FDR control procedures can be generalized to a broad class of binary probabilistic classification problem for prioritizing the control of one specific type of misclassification errors, i.e., the false positives.
From this viewpoint, hypothesis testing is simply a special case with null hypotheses being framed as non-discoveries, and as a generalization, applying Bayesian procedures help overcome some practical limitations associated with formulating and solving hypothesis testing problems. 
Here we demonstrate this point in two cases.

\paragraph{Case 1: testing complex null hypothesis} \mbox{} In standard hypothesis testing practice, a key step is to set up a null hypothesis with a well-defined $f_0$ distribution. With the emerging complex data structures in applications, this can be difficult as some ideal null models may involve very complex scenarios. 
An example of this kind is to test colocalization in genetic association studies where the scientific goal is to statistically assess if a genetic variant is simultaneously associated two separate complex traits. Let latent binary indicators $\gamma_i$  and $\delta_i$ denote the association status of the $i$-th genetic variant with trait 1 and trait 2, respectively.  
In the formulation of hypothesis testing problem, the goal is to test
\begin{equation*}
  \begin{aligned}
       &  H_0:  \gamma_i \ne 1 \mbox{  or  } \delta_i \ne 1, \\ 
       & ~~~~~~~~~~~~ vs.  \\
       & H_1:  \gamma_i = 1 \mbox{ and } \delta_i = 1.
   \end{aligned}
\end{equation*}
Although the marginal associations with respect to $\gamma_i$ and $\delta_i$ are relatively easy to test,  the null hypothesis for colocalization involves multiple compatible and mutually exclusive scenarios.
Consequently, the corresponding $f_0$ distribution is a mixture distribution for which $p$-value is difficult to compute or calibrate.
In comparison, the Bayesian approach only needs to focus on evaluating $\Pr(\gamma_i = 1, \delta_i = 1 \mid \mbox{ data})$ for the alternative scenario,  and  the quantity, $1  - \Pr(\gamma_i = 1, \delta_i = 1 \mid \mbox{ data})$, is indeed the desired local fdr.  (A practical Bayesian solution for the colocalization problem can be found in \cite{Wen2017}.)
Another related example is the estimation of the empirical null distribution in the local fdr framework \cite{EfronBook}. Similarly, such idea is extremely difficult, if not impossible, to implement in the frequentist FDR control framework.  

\paragraph{Case 2: flexible definition of discovery} \mbox{} The Bayesian framework allows the flexible definition of discoveries in probabilistic terms, which can lead to more interpretable results for better scientific communication. 
A longstanding scientific problem that motivates the development of FDR control technique is the analysis of microarray and RNA-seq data analysis for identifying differentially expressed genes in two experimental conditions.
The common practice in this context is to test a ``sharp null" hypothesis which states that the difference of the expression levels of a target gene from the two conditions is exactly 0.  
Nevertheless, many have pointed out the caveat of the sharp null hypothesis (e.g., it is scientifically plausible that all genes are differentially expressed under the different conditions) and argued that the real purpose of the analysis is to identify genes with the convincingly large non-zero difference in expression levels, denoted by $\beta$.
Recently, \cite{Stephens2016} derives the local false sign rate (lfsr) as the posterior error probability in estimating the sign of $\beta$, which enables probabilistic classification of large versus small effects and a straightforward application of Bayesian FDR control procedure. 
In another example, \cite{Li2011} defines a discovery in high-throughput genomic experiments if a noticeable observation is replicable in a follow-up validation study. 
They derive a probabilistic measure, irreproducible discovery rate (IDR), and apply the Bayesian FDR control procedure to identify those convincingly reproducible signals.
In both examples, there seem no straightforward applications of the frequentist FDR control procedure.

Finally, we conclude that our work in this paper should motivate more applications of Bayesian FDR control approaches involving large-scale and complex data: the Bayesian approaches can be flexible, powerful, robust, and most importantly, verifiable by examining their frequentist properties. 

\section{Resources}

The code implementing the illustrations and simulations used in this paper can be found in the github repo: \url{https://github.com/xqwen/unified_fdr}.

\newpage

\begin{appendix}
\section{Proof of Theorem 1}
\begin{proof}
By equation (\ref{bfdr}), 
\begin{equation}
   \BFDR \left (t_{b,i} =  u^*_{(i)} \right) = \frac{\left( \sum_{k=1}^m u_{k}^*  \cdot  \lv\{ u^*_k \le u^*_{(i)}\}  \right) \big/ m }{i/m}.
\end{equation}
Note that, individually, because $u_k^*(z_k) = {\rm E}( \lv \{ \gamma_k = 0 \} \mid \pi_0, z_k)$, it follows that
\begin{equation}\label{exp.condition}
  {\rm E} \left[  u_k^* \cdot \lv{\{ u_k^* \le  u_{(i)}^* \} }\,  \bigl\vert\, \pi_0, u_{(i)}^* \right]  = {\rm E} \left[ \lv{\{ u_k^* \le  u_{(i)}^* \mbox{ and }  \gamma_k = 0 \}} \, \bigl\vert\, \pi_0, u_{(i)}^* \right] = \pi_0 \, F_{u^*} ( u_{(i)}^* ), \forall ~ k, 
\end{equation}
% expectation is w.r.t $z_k$, the form is based on the total expectation law
where $F_{u^*}(\cdot)$ denotes the CDF of  $u^*$'s when data are generated from $f_0$.  
It then follows directly from the strong law of large number (SLLN),  
\begin{equation*}
 \frac{\left( \sum_{k=1}^m u_{k}^*  \cdot  \lv\{ u^*_k \le u^*_{(i)}\}  \right) \big/ m }{i/m}   \xrightarrow{a.s.} \frac{m \pi_0 F_{u^*}(u^*_{(i)})}{i}~ \mbox{ for all  $i$}.
\end{equation*}
Note,  Equation (\ref{bf2pip}) indicates that $u^*(z)$ is a monotone function of $\BF^*(z)$, and implies that there exists a monotone mapping, denoted by $G(\cdot)$, between  $u^*$ and the corresponding $p$-value by the assumption 2.  As a consequence,  $G(u^*_{(i)}) = p_{(i)}$.  
It follows that 
\begin{equation*}
 F_{u^*}(u^*_{(i)}) = \Pr \left[ u^* \le u^*_{(i)} \mid H_0 \right] = \Pr \left[ G(u^*) \le G(u^*_{(i)}) \mid H_0 \right] = p_{(i)}.
\end{equation*}

Finally, by equation (\ref{ffdr.exp}) we conclude that
\begin{equation*}
  \BFDR \left (t_{b,i} =  u^*_{(i)} \right)  \xrightarrow{a.s.} \FDR \left(t_{f,i} =  p_{(i)} \right) ~\mbox{for all $i$}.
\end{equation*}
\end{proof}
 
\newpage

\section{Proof of Theorem 2}

\begin{proof}
The proof is built upon the same technical arguments used in \cite{Sun2007} (SC in brief, henceforward). 

We first define
\begin{equation*}
  \tilde u (z)  := \frac{\hat \pi_0 f_0 (z) }{f_c(z)} = {\rm E}\left( \lv\{\gamma_i=0\} \mid \hat \pi_0, z  \right).
\end{equation*}

Given that ${\rm E} || \hat f_c - f_c ||^2 \to 0$, it can be shown, by Lemma A.1 and Lemma A.2 of SC, that
\begin{equation} \label{pip.conv}
  \hat u \xrightarrow{p} \tilde u, \forall i,
\end{equation}
and consequently, for any given order statistics from observed sequences
\begin{equation} \label{pip_os.conv}
  \hat u_{(k)} \xrightarrow{p} \tilde u_{(k)}, \forall k,
\end{equation}
%Note that,
%\begin{equation} \label{pip.exp}
% {\rm E} \left[  \tilde u_k \cdot \lv{\{ \tilde u_k \le  \tilde u_{(i)} \} }\,  \bigl\vert\, \hat \pi_0 \right]  = \hat \pi_0 \, F_{u^*} ( u_{(i)}^* ), \forall ~ k. 
%\end{equation}
Equation (\ref{pip.conv}) and the Helly-Bray theorem imply
 \begin{equation}
   g(t) = {\rm E} \left[  \hat u_k \cdot \lv{\{ \hat u_k \le  t \} } \,  \bigl\vert\, \hat \pi_0, t \right]  \to    {\rm E} \left[  \tilde u_k \cdot \lv{\{ \tilde u_k \le  t \} }\,  \bigl\vert\, \hat \pi_0, t \right] = h(t) , ~ \forall  t \in [0,1] .
\end{equation} 
Furthermore  by Equation (\ref{pip_os.conv}), $ g( \hat u_{(k)} ) \xrightarrow{p} h(\tilde u_{(k)}) $, i.e.,
\begin{equation}\label{exp.conv}
   {\rm E} \left[  \hat u_k \cdot \lv{\{ \hat u_k \le  \hat u_{(i)} \} }\,  \bigl\vert\, \hat \pi_0, \hat u_{(i)} \right] \xrightarrow{p}   {\rm E} \left[  \tilde u_k \cdot \lv{\{ \tilde u_k \le  \tilde u_{(i)} \} }\,  \bigl\vert\, \hat \pi_0, \tilde u_{(i)} \right] , ~ \forall ~ k.  % \hat \pi_0 \, F_{ \tilde u} ( \tilde u_{(i)} ), \forall ~ k.
\end{equation} 

Next, we show that 
\begin{equation} \label{hat.exp}
 {\rm E} \left[  \hat u_k \cdot \lv{\{ \hat u_k \le  \hat u_{(i)} \} }\,  \bigl\vert\, \hat \pi_0, \hat u_{(i)} \right] \xrightarrow{p} \hat \pi_0 p_{(i)}.  
\end{equation}
This is because 
\begin{equation}\label{tilde.exp}
\begin{aligned}
   {\rm E} \left[  \tilde u_k \cdot \lv{\{ \tilde u_k \le  \tilde u_{(i)} \} }\,  \bigl\vert\, \hat \pi_0, \tilde u_{(i)} \right]  & = {\rm E} \left( \lv\{ \tilde u_k \le \tilde u_{(i)} \mbox{  and } \gamma_k = 0 \} \mid \hat \pi_0, \tilde u_{(i)} \right) \\
      & = \hat \pi_0 F_{\tilde u}\left( \tilde u_{(i)} \right),
  \end{aligned}    
\end{equation}
where $ F_{\tilde u}\left( \cdot \right) $ denotes the cdf of $\tilde u (z) $, and similarly we denote $F_{\hat u} \left(\cdot \right)$ as the cdf of  $\hat u (z)$ under the null.
Note that, by (\ref{pip.conv}) and (\ref{pip_os.conv}), 
\begin{equation} \label{cdf.conv}
  F_{\tilde u}\left( \tilde u_{(i)} \right) = {\rm E} \left( \lv\{ \tilde u \le \tilde u_{(i)} \} \mid H_0,  \tilde u_{(i)} \right) \xrightarrow{p} {\rm E} \left( \lv\{ \hat u \le \hat u_{(i)} \} \mid H_0,  \tilde u_{(i)} \right) = F_{\hat u} \left(\hat u_{(i)} \right)
\end{equation}
Under the assumption $\BF(z) = \frac{f_0(z)}{\hat f_1(z)}$ is monotonic to the corresponding $p$-value,  it follows that for any given $\hat \pi_0$, there exists a monotone transformation, $\hat G(\cdot)$, between $\hat u$ and its corresponding $p$-value, such that $ \hat G(\hat u) = p$, and $\hat G(\hat u_{(i)}) = p_{(i)}$ for any given order statistics.  
It follows that,
 \begin{equation} \label{pval.eq}
     F_{\hat u}\left( \hat u_{(i)} \right) = \Pr \left[ \hat u \le \hat u_{(i)} \mid H_0 \right] =  \Pr \left[ \hat G(\hat u) \le \hat G(\hat u_{(i)}) \mid H_0 \right] = p_{(i)}. 
\end{equation}

Jointly by (\ref{exp.conv}), (\ref{tilde.exp}), (\ref{cdf.conv}) and (\ref{pval.eq}), we have established (\ref{hat.exp}), i.e.,
$${\rm E} \left[  \hat u_k \cdot \lv{\{ \hat u_k \le  \hat u_{(i)} \} }\,  \bigl\vert\, \hat \pi_0, \hat u_{(i)} \right] \xrightarrow{p} \hat \pi_0 p_{(i)}. $$
Finally, note that 
\begin{equation}
   {\rm BFDR}\left( t_{b,i} = \hat u_{(i)}; \hat \pi_0 \right)  \xrightarrow{a.s.}  \frac{{\rm E} \left[  \hat u_k \cdot \lv{\{ \hat u_k \le  \hat u_{(i)} \} }\,  \bigl\vert\, \hat \pi_0 \right]}{i/m},
\end{equation}
and 
\begin{equation}
   {\rm FDR}\left( t_{f,i} = p_{(i)}; \hat \pi_0 \right)  \xrightarrow{a.s.}  \frac{\hat \pi_0 p_{(i)}}{i/m},
\end{equation}
we therefore conclude that 
  \begin{equation}
     {\rm BFDR}\left( t_{b,i} = \hat u_{(i)}; \hat \pi_0 \right) \xrightarrow{p} {\rm FDR}\left( t_{f,i} = p_{(i)}; \hat \pi_0 \right), \forall i.
  \end{equation}

In addition, if the condition $\hat \pi_0 \xrightarrow{p} \pi_0$ is satisfied, i.e., $\pi_0$ is accurately estimated by $\hat \pi_0$.
It is straightforward to show that
\begin{equation}
    \BFDR \left (t_{b,i} =  \hat u_{(i)} ; \hat \pi_0 \right)  \xrightarrow{p}  \frac{\pi_0 p_{(i)}}{i/m}.
\end{equation}
By Theorem 1, it is sufficient to conclude 
\begin{equation}
     \BFDR \left (t_{b,i} =  \hat u_{(i)} ; \hat \pi_0 \right)  \xrightarrow{p} \BFDR(t_{b,i} = u^*_{(i)}), \forall i.
\end{equation}

% 
%The existence of monotone mapping between $\BF(z_k) := f_0(z_k)/\hat f_1 (z_k) $ and the corresponding $p$-value indicates that for any given $\hat \pi_0$,  there exists a monotone mapping, $\hat G(\cdot)$, between the sequences $\{\hat u_{(k)}\}$ and $\{p_{(k)}\}$, such that
%\begin{equation}
%    \hat G \left(\hat u_{(k)}\right) = p_{(k)}.
%\end{equation}
%Along with Equation (\ref{pip.conv}), it implies
%

\end{proof}

\newpage

\section{Proof of Proposition 1}
\begin{proof}
Note that 
\begin{equation*}
 \BFDR \left (t_{b,i} =  u^*_{(i)} \right) = \frac{\left( \sum_{k=1}^m u_{k}^*  \cdot  \lv\{ u^*_k \le u^*_{(i)}\}  \right) \big/ m }{i/m}.
\end{equation*}
Individually, 
\begin{equation*}
  {\rm E} \left[  u_k^* \cdot \lv_{\{ u_k^* \le  u_{(i)}^* \} }\,  \bigl\vert \, d_i \right]  = {\rm E} \left[ \lv_{\{wlr^*_k \ge  wlr_{(i)}^*  \mbox{ and }  \gamma_i = 0 \}} \, \bigl\vert\,  d_i \right] = \pi_{d_i,0} \, [1- F_{d_i, wlr^*} ( u_{(i)}^* )], 
\end{equation*}
for all $k$, and
\begin{equation*}
 {\rm Var}\left[ \lv_{\{wlr^*_k \ge  wlr_{(i)}^*  \mbox{ and }  \gamma_i = 0 \}} \, \bigl \vert\,  d_i \right]  \mbox{ is obviously bounded.}
\end{equation*}
It follows from the Kolmogorov's SSLN that
 \begin{equation*}
    \BFDR\left ( t_{b,i} = u^*_{(i)} \right )  \xrightarrow{a.s.}  \FDR \left (t_{f,i} = wlr_{(i)}^* \right) \mbox{ for all } i.
\end{equation*}

\end{proof}

\end{appendix}

\newpage

\bibliography{unified_fdr.bib}

\end{document}